\documentclass[11pt]{article}

\usepackage[T1]{fontenc}
\usepackage[utf8]{inputenc}
\usepackage[english]{babel}
\usepackage{amsmath,amssymb,amsthm,mathtools}
\usepackage{booktabs,array,float}
\usepackage[protrusion=false]{microtype}
\usepackage[letterpaper,hmargin=1.005in,vmargin=1.05in,includefoot]{geometry}
\usepackage[numbers]{natbib}
\usepackage[hidelinks]{hyperref}
\providecommand{\doi}[1]{\href{https://doi.org/#1}{\nolinkurl{https://doi.org/#1}}}
\hypersetup{
  pdflang={en-US},
  pdftitle={Coding for Multiple Reverse-Complement and Palindromic Duplications},
  pdfauthor={Aryeh Lev Zabokritskiy (Yohananov)},
  pdfsubject={Finite-error duplication-correcting codes},
  pdfkeywords={reverse-complement duplication, palindromic duplication, duplication-correcting codes, DNA storage, sphere-packing bounds}
}

\newtheorem{theorem}{Theorem}[section]
\newtheorem{lemma}[theorem]{Lemma}
\newtheorem{proposition}[theorem]{Proposition}
\newtheorem{corollary}[theorem]{Corollary}
\theoremstyle{definition}
\newtheorem{definition}[theorem]{Definition}

\newtheorem{openproblem}[theorem]{Open Problem}
\theoremstyle{remark}
\newtheorem{remark}[theorem]{Remark}

\newcommand{\rc}{\mathrm{rc}}
\newcommand{\pal}{\mathrm{pal}}
\newcommand{\Amap}{\mathsf{A}}
\newcommand{\cU}{\mathcal{U}}
\newcommand{\Z}{\mathbb{Z}}
\newcommand{\eqdef}{\triangleq}
\DeclareMathOperator{\InsAfter}{InsAfter}
\DeclareMathOperator{\weave}{weave}
\DeclareMathOperator{\Dec}{Dec}

\title{Coding for Multiple Reverse-Complement and Palindromic Duplications}
\author{Aryeh Lev Zabokritskiy (Yohananov)\thanks{\small Department of Computer
Science, Tel-Hai University of Kiryat Shmona in the Galilee; MIGAL --
Galilee Research Institute.  Email:
\href{mailto:yuhanalev@telhai.ac.il}{yuhanalev@telhai.ac.il}.  ORCID:
\href{https://orcid.org/0000-0003-3151-6192}{0000-0003-3151-6192}.}}
\date{}

\begin{document}
\maketitle

\begin{abstract}
Reverse-complement (RC) and palindromic (PAL) duplications insert a reversed copy of a length-$k$
block immediately after the block; an RC
duplication also applies a fixed involutive complement to the inserted symbols.  We study $q$-ary codes
  correcting $t$ such operations applied in sequence, where later operations
  may copy inserted symbols.  For fixed $q\geq2$ and $k,t\geq1$, we prove that every
  length-$n$ code $C$ for either channel has redundancy
  $n-\log_q|C|\geq t\log_q n-O_{q,k,t}(1)$; for fixed $q,k$ and $1\leq t=t(n)=o(n)$ the lower bound is
$t\log_q(n/t)-O_{q,k}(t)$.  For fixed $k$, known codes for one length-$k$
insertion burst attain redundancy $\log_q n+O_{q,k}(1)$ for a single
duplication of either type.  For every
even $k$ and any involutive complement, we show that the two channels are equivalent for every number of errors: the bijection complementing every other coordinate turns each RC duplication
into a PAL duplication.  This bijection therefore
converts any $t$-error-correcting RC code into a PAL code of the same size, and
conversely, with linear-time transfer of encoders and decoders.  For every even $k$, we also determine the maximum number of
descendants of one word under exactly two errors; by the bijection,
it is the same for both channels.  For fixed even $k$, a direct count using this
maximum proves the existence of two-error-correcting PAL and RC codes with
redundancy at most $4\log_q n+O_{q,k}(1)$.  A known existence bound for two deletion bursts transfers
to insertions with the same leading term as ours and a larger additive constant.
For two errors, the converse gives
$2\log_q n-O_{q,k}(1)$, a factor of two below these existence bounds.
\end{abstract}

\begingroup
\par
\endgroup

\section{Introduction}
\label{sec:introduction}

String-duplication systems were developed as models of the generation of repeated
sequences~\cite{FarnoudSchwartzBruck2016}.  The standard tandem-duplication
channel inserts an unchanged copy of a factor immediately after it
~\cite{JainFarnoudSchwartzBruck2017}.  For a fixed
tandem-duplication length and a prescribed number of errors, bounds and asymptotically
optimal constructions are known
~\cite{KovacevicTan2018,LenzJungerWachterZeh2018}.  The length-one
specialization is closely related to sticky insertions, for which efficient
asymptotically optimal codes are also known~\cite{MahdavifarVardy2017}.

Here we study two channels in which the copied factor is reversed.
A length-$k$ palindromic (PAL) duplication~\cite{FarnoudSchwartzBruck2016,LenzWachterZehYaakobi2019}
copies a factor, reverses it, and inserts the reversed copy immediately after
the source factor.  A reverse-complement (RC) duplication~\cite{BenTolilaSchwartz2022} additionally
complements every copied symbol.  In a sequence of duplications, each operation
acts on the current word, so the factor it copies may include symbols
inserted by earlier operations.

We work over an alphabet $\Sigma$ of size $q\geq2$ with an involutive
complement map, which may have fixed points.
Fix integers $n\geq k\geq1$ and $t\geq0$, and a channel
$\chi\in\{\pal,\rc\}$.  A code
$C\subseteq\Sigma^n$ corrects exactly $t$ length-$k$ duplications in channel
$\chi$ when distinct codewords have no common descendant after $t$
operations.  Its $q$-ary redundancy is
  $r_q(C)\eqdef n-\log_q|C|$.  Exact-$t$ and at-most-$t$
  correction are equivalent for these insertion-only channels.  The
finite-error problem asks for the minimum redundancy under either convention.
The symbol $C$ always denotes a code for the channel under discussion.

For tandem and palindromic duplication, Lenz, Wachter-Zeh, and
Yaakobi established several bounds and one-error constructions.  Their
conclusion explicitly listed ``code size upper bounds for multiple
palindromic duplications'' as an open direction
~\cite{LenzWachterZehYaakobi2019}.  Their model fixes the duplication
length, permits any legal start at each successive step, and defines correction
through balls of at most $t$ errors.  Since the two conventions coincide, their palindromic model is the PAL
channel studied here.  We give a general finite-length upper bound in this direction.  For fixed $t$,
it yields
  $r_q(C)\geq t\log_q n-O_{q,k,t}(1)$ for PAL and RC errors of every length $k$.  Yohananov and Schwartz later determined the unbounded-error
coding capacities of both channels over even alphabets, for RC with a
fixed-point-free complement~\cite{YohananovSchwartz2025}; for every $k\geq2$,
both capacities are zero~\cite[Corollaries~13 and~16]{YohananovSchwartz2025}.  That result leaves open the finite-error regime,
in which codes with sublinear redundancy can exist.

The closest parity-sensitive predecessor is the work of Ben-Tolila and
Schwartz on one RC duplication of odd fixed length over an even alphabet with
a fixed-point-free complement~\cite{BenTolilaSchwartz2022}.  They prove a
one-error sphere-packing converse~\cite[Theorem~18]{BenTolilaSchwartz2022},
slightly sharper in their setting than our bound at $t=1$, and give an
odd-$k$ construction~\cite[Theorem~17]{BenTolilaSchwartz2022}: a binary code for
one arbitrary length-$k$ insertion burst, which is a stronger error, becomes a
$q$-ary RC code.  Instantiating that construction with the earlier burst code of
Schoeny, Wachter-Zeh, Gabrys, and Yaakobi gives asymptotically optimal redundancy
with an additional $O_k(\log\log n)$ term for fixed odd $k\geq3$
~\cite{SchoenyWachterZehGabrysYaakobi2017}.  Thus this RC-specific construction
covers one error at odd $k$; it neither covers multiple errors nor extends
to even $k$.

For one error, however, a code correcting one arbitrary length-$k$ insertion burst
also corrects one PAL or RC duplication, and recent burst codes achieve
$\log_q n+O_{q,k}(1)$ redundancy for every fixed $k$
~\cite{SunLuZhangGe2025}.  For any number of RC or PAL
duplications of pairwise disjoint factors of the original word, all of a common length at least $3\lceil\log_q n\rceil$, Sun and Ge give
a code with one redundant symbol~\cite{SunGe2026}.  Two other RC results
assume even $q\geq4$ and a fixed-point-free complement.  For a fixed number $t$
of length-one duplications, Sun and Ge give codes with redundancy
$(2t-1)\log_q n+O(\log_q\log_q n)$ and efficient encoding and decoding
~\cite{SunGe2026}; Liu, Tang, Fan, and Sagar give codes correcting a single RC
duplication of unknown length $k\geq2$ with redundancy at most
$n(1-\log_q(q-1))$~\cite{LiuTangFanSagar2026}.

The odd-$k$ construction takes the preimage of a binary code under a
complement-preserving map from the original alphabet to a binary alphabet.  Sun and Ge show that this preimage
can fail already over $q=4$: it need not correct one RC duplication of length
$k=2$, and it can also fail for two RC duplications of length one
~\cite[Example~V.3]{SunGe2026}.  Our alternating-complement map for even $k$ instead acts as a
coordinate-dependent bijection of the full word space: under zero-based indexing it
complements exactly the odd coordinates, is an involution, and conjugates each
reverse-complement step to a palindromic step.  Unlike the odd-$k$ construction, this
conjugacy requires only an involution on $\Sigma$, possibly with fixed points,
  for every $q\geq2$.  Iterating the conjugacy therefore covers arbitrary
  sequential histories and hence every number of errors.

Our main results are the following.
\begin{enumerate}
  \item \emph{General finite-error converse.}
  For either channel, every code correcting $t$
  duplications satisfies, for fixed $q\geq2$ and $k,t\geq1$,
  \[
    |C|=O_{q,k,t}\!\left(\frac{q^n}{n^t}\right),
    \qquad
    r_q(C)\geq t\log_q n-O_{q,k,t}(1).
  \]
  For fixed $q,k$ and integer-valued $1\leq t=t(n)=o(n)$, the bound becomes
  $t\log_q(n/t)-O_{q,k}(t)$.

  \item \emph{Even-length conjugacy.}
  At every even $k$ and for any involutive complement, the alternating-complement map is a bijective conjugacy
  of the full sequential channels.  Applying it to every word maps any
  $t$-error-correcting PAL code to an RC code with the same length and size,
  and conversely, so the two channels have equal optimal code sizes.  An encoder or decoder transfers by adding at most two
  coordinate passes, each linear in the word length.

  \item \emph{Maximum exact-two spheres at even duplication length.}
  For every even $k$, the maximum exact-two descendant sphere has size
  \[
    (n-k+1)(n+1)-\binom{\max\{n-2k+2,0\}}{2}.
  \]
  Words alternating between any two distinct alphabet symbols attain it for
  PAL, and conjugacy gives the corresponding RC maximizers.

  \item \emph{Direct two-error existence bound.}
  For every even $k$, a direct count using this sphere maximum gives a
  finite-length lower bound on the
  maximum size of two-error-correcting PAL and RC codes; for fixed $k$, it yields
  redundancy at most $4\log_q n+O_{q,k}(1)$.
  The two-burst deletion--insertion equivalence is stated by
  Wang et al.~\cite[Sec.~II]{WangKongYaakobiDuman2026}; we prove the direction
  needed to transfer the $q$-ary existence bound of Ye et al.\ for two bursts
  of exactly $k$ deletions~\cite[Theorem~III.1 and Corollary~III.1]{YeSunYuGeElishco2026}
  to insertions.
  The transferred bound gives the same leading order for every $k\geq2$, with an additive term
  larger than ours by $2k-\log_q2+o(1)$ at even $k$.
  For fixed $q\geq2$ and fixed even $k$, the converse at $t=2$ and the
  direct existence bound thus bracket the minimum redundancy of
  two-error-correcting codes for each channel between $2\log_q n-O_{q,k}(1)$ and
  $4\log_q n+O_{q,k}(1)$.
\end{enumerate}

For one error, for example, the conjugacy transfers the binary length-two PAL code of
Lenz, Wachter-Zeh, and Yaakobi to RC, with redundancy at most $\log_2 n+O(1)$.
For an algorithmic comparison, we adapt the binary syndrome construction of Ye et al.\
for two deletion bursts~\cite[Theorem~IV.1]{YeSunYuGeElishco2026}
to two arbitrary insertion bursts of fixed length $k\geq2$, and hence to
two PAL or RC duplications.  The resulting redundancy upper bound has leading
coefficient five.  A suitable syndrome tuple exists, and once it is supplied,
the code has a polynomial-time codeword
decoder.  The construction gives no efficient parameter selection or message encoder.

Section~\ref{sec:model} defines the channel and its coding parameters.
Section~\ref{sec:separator-subcone-packing} proves the general packing
bound.  Section~\ref{sec:even-equivalence} develops the even-length
conjugacy and its algorithmic transfer.  Sections~\ref{sec:edge-tree},
\ref{sec:existence}, and \ref{sec:target} treat the two-error sphere
geometry, the direct existence bound, and the burst-code benchmarks and open binary coding problems.

\section{Channel model and code parameters}
\label{sec:model}

Let $\Sigma$ be a finite alphabet of size $q\geq2$, equipped with an
involution $\iota:\Sigma\to\Sigma$, written
$\overline a\eqdef\iota(a)$.  In the usual DNA model the involution has no fixed
points; the general results below also permit fixed points.  In particular,
the general $q$-ary results include the identity involution.
Earlier work
often writes the alphabet as $\Z_q$; taking $\Sigma=\Z_q$, with any involution, recovers that
convention.  In the binary RC
specializations below, $\Sigma=\{0,1\}$ and
$\overline a=1-a$.  We
extend the complement coordinate-wise to words and index all coordinates
from zero.  For $x=x_0\cdots x_{n-1}$, its reversal is
$x^R\eqdef x_{n-1}\cdots x_0$; for $m\geq0$, the power $x^m$ denotes the
concatenation of $m$ copies of $x$.

If $x=uvw\in\Sigma^n$, where $|u|=i$ and $|v|=k$, define
\begin{align}
  T^{\rc}_{i,k}(x)&\eqdef uv\,(\overline v)^R w
  =uv\,\overline{v^R}w,\label{eq:rc-duplication}\\
  T^{\pal}_{i,k}(x)&\eqdef uv\,v^Rw.\label{eq:pal-duplication}
\end{align}
The first operation inserts a reverse-complement copy of $v$, while the
second inserts a reversed copy.  For a channel label
$\chi\in\{\pal,\rc\}$, the operation $T^{\chi}_{i,k}$ inserts
$\iota_\chi(v^R)$, where
\[
  \iota_{\pal}(a)\eqdef a,
  \qquad
  \iota_{\rc}(a)\eqdef\iota(a),
\]
and each map is extended coordinate-wise to words.

For $\chi\in\{\pal,\rc\}$, let $D^t_{k,\chi}(x)$ be the set of
descendants obtained from $x$ by a sequence of exactly $t$ length-$k$ operations
of type $\chi$, each acting on the current word, with $D^0_{k,\chi}(x)=\{x\}$.  Put
\[
  D^{\leq t}_{k,\chi}(x)
  \eqdef \bigcup_{s=0}^{t}D^s_{k,\chi}(x).
\]
We call $D^t_{k,\chi}(x)$ the \emph{exact-$t$ descendant sphere} of $x$
and $D^{\leq t}_{k,\chi}(x)$ its \emph{at-most-$t$ descendant ball}.
A code $C\subseteq\Sigma^n$ corrects exactly $t$ such duplications when
the sets $D^t_{k,\chi}(x)$ are pairwise disjoint for $x\in C$; it
corrects at most $t$ duplications when the sets
$D^{\leq t}_{k,\chi}(x)$ are pairwise disjoint.  Recall that its $q$-ary
redundancy is
\[
  r_q(C)=n-\log_q|C|.
\]

For either the exact-$t$ or at-most-$t$ convention, the associated
\emph{confusability graph} has vertex set $\Sigma^n$ and joins two distinct
words when their corresponding descendant sets intersect.  Its independent
sets are precisely the correcting codes for that convention.

The two correction conventions agree for these insertion-only channels.

\begin{lemma}[Exact and at-most correction]
\label{lem:exact-at-most}
Fix $\chi\in\{\pal,\rc\}$, $1\leq k\leq n$, and $t\geq0$.  A code corrects
exactly $t$ length-$k$ duplications of type $\chi$ if
and only if it corrects at most $t$ such duplications.
\end{lemma}

\begin{proof}
One implication is immediate.  Conversely, suppose two distinct codewords
have a common descendant after exactly $s<t$ operations.  Starting from
that common word, apply any $t-s$ further legal duplications.
The resulting word is a common exact-$t$ descendant, a contradiction.
Collisions between different error counts cannot occur because the output
length after $s$ operations is $n+sk$.
\end{proof}

\section{An endpoint-multiplicity packing bound}
\label{sec:separator-subcone-packing}

We prove the general finite-error converse as a packing bound, selecting a
subfamily of explicitly known size inside each exact-$t$ descendant sphere.  Recall that a length-$k$ operation in
channel $\chi\in\{\pal,\rc\}$ inserts $\iota_\chi(v^R)$ immediately
after its source factor $v$, where $\iota_\chi$ is the symbol map
specified by the channel model.  Throughout this section,
$q=|\Sigma|\geq2$.  By Lemma~\ref{lem:exact-at-most}, every bound below for
codes with pairwise disjoint exact-$t$ descendant spheres also holds for
codes correcting at most $t$ duplications.

\begin{definition}[Eligible endpoints]
\label{def:eligible-endpoints}
For $x=x_0\cdots x_{n-1}\in\Sigma^n$ and $1\leq k\leq n$, call $j\in\{k-1,\ldots,n-1\}$
the endpoint of the length-$k$ source block $x_{j-k+1}\cdots x_j$, and define
\[
  J_{k,\chi}(x)
  \eqdef
  \{n-1\}
  \cup
  \bigl\{j\in\Z:
    k-1\leq j\leq n-2,\
    x_{j+1}\ne\iota_\chi(x_j)
  \bigr\},
\]
and put
\[
  G_{k,\chi}(x)\eqdef |J_{k,\chi}(x)|.
\]
The terminal endpoint $n-1$ is always eligible, so
$G_{k,\chi}(x)\geq1$.
\end{definition}

The inequality $x_{j+1}\ne\iota_\chi(x_j)$ in
Definition~\ref{def:eligible-endpoints} makes the number of copies inserted after each eligible endpoint recoverable from
the resulting word when the source $x$ is known.  A transformed reversed copy of a source block ending at $j$ begins
with $\iota_\chi(x_j)$, whereas the unmodified source continues with
$x_{j+1}$.  At an eligible nonterminal endpoint, these two symbols therefore
distinguish another inserted copy from forward progress in the source.

\begin{lemma}[Endpoint-multiplicity embedding]
\label{lem:separator-subcone}
For every channel label $\chi\in\{\pal,\rc\}$, $x\in\Sigma^n$,
$1\leq k\leq n$, and $t\geq0$,
\begin{equation}
  \bigl|D^t_{k,\chi}(x)\bigr|
  \geq
  \binom{G_{k,\chi}(x)+t-1}{t}.
  \label{eq:separator-subcone-size}
\end{equation}
\end{lemma}

\begin{proof}
For every $j\in\{k-1,\ldots,n-1\}$, let
\[
  v_j\eqdef x_{j-k+1}\cdots x_j,
  \qquad
  w_j\eqdef\iota_\chi(v_j^R).
\]
Choose nonnegative multiplicities
$\mathbf m=(m_j)_{j\in J_{k,\chi}(x)}$ with
$\sum_jm_j=t$, and set $m_j\eqdef0$ outside $J_{k,\chi}(x)$.  Define
\[
  F^\chi_{\mathbf m}(x)
  \eqdef
  x_0\cdots x_{k-2}
  \prod_{j=k-1}^{n-1}\bigl(x_jw_j^{m_j}\bigr),
\]
where the initial prefix is empty when $k=1$, and the product denotes
concatenation in increasing order of $j$.  We call these words the
endpoint-multiplicity family of $x$.

This word is an exact-$t$ descendant.  Process the eligible endpoints in
strictly decreasing order.  When endpoint $j$ is reached, every operation
already performed inserted symbols to the right of an original endpoint
larger than $j$.  The original occurrence
$v_j=x_{j-k+1}\cdots x_j$ is therefore still contiguous.  Duplicate this
same occurrence $m_j$ times.  It remains intact after each operation, and
each repetition inserts the same word $w_j$ immediately after the original
$x_j$.  Operations at smaller endpoints only push the completed suffix to
the right.  This proves reachability of $F^\chi_{\mathbf m}(x)$, including
when the selected source blocks overlap.

It remains to prove injectivity.  Suppose that
$\mathbf m\ne\mathbf m'$ are two weak compositions of the same $t$, and
let $j$ be the smallest eligible endpoint at which their multiplicities
differ.  The endpoint cannot be terminal if all earlier multiplicities
agree, because both total sums equal $t$.  Assume without loss of generality
that $m_j<m'_j$.  After the common prefix and the first $m_j$ copies of
$w_j$, the word $F^\chi_{\mathbf m}(x)$ continues with the original symbol
$x_{j+1}$, whereas $F^\chi_{\mathbf m'}(x)$ continues with the first symbol
of another $w_j$, namely $\iota_\chi(x_j)$.  These symbols are different
by eligibility.  Hence the descendants are distinct, and counting weak
compositions proves \eqref{eq:separator-subcone-size}.
\end{proof}

To turn Lemma~\ref{lem:separator-subcone} into a packing bound, it remains to
control the number of ambient words with too few eligible endpoints.  We
count them with the uniform distribution on $\Sigma^n$: the number of
ambient words in an event is $q^n$ times its probability.
A lower-tail estimate for $G_{k,\chi}$ controls these words,
while Lemma~\ref{lem:separator-subcone} supplies a large sphere for every
remaining word.  Disjointness then packs those spheres inside
$\Sigma^{n+kt}$.

Let $X=(X_0,\ldots,X_{n-1})$ be uniform on $\Sigma^n$, and put
\begin{equation}
  N\eqdef n-k,
  \qquad
  p\eqdef\frac{q-1}{q}.
  \label{eq:packing-parameters}
\end{equation}
For either channel,
\begin{equation}
  G_{k,\chi}(X)-1
  =
  \sum_{j=k-1}^{n-2}
  \mathbf 1\{X_{j+1}\ne\iota_\chi(X_j)\}
  \sim\operatorname{Bin}(N,p).
  \label{eq:eligible-endpoint-binomial}
\end{equation}
Indeed, the prefix preceding $X_{k-1}$ is arbitrary and cancels from the
probability.  For any prescribed pattern of equalities and inequalities
between $X_{j+1}$ and $\iota_\chi(X_j)$ on the remaining $N$ transitions,
$X_{k-1}$ has $q$ choices, every constrained equality has one
choice, and every constrained inequality has $q-1$ choices.  The transition
indicators are therefore independent Bernoulli variables with parameter
$p$.

\begin{theorem}[Finite-length packing bound]
\label{thm:finite-length-packing}
Let $1\leq k\leq n$, $t\geq1$, and
$\chi\in\{\pal,\rc\}$.  Suppose that
$C\subseteq\Sigma^n$ has pairwise disjoint exact-$t$ descendant spheres:
\[
  D^t_{k,\chi}(x)\cap D^t_{k,\chi}(x')=\varnothing
  \qquad
  \text{for all distinct }x,x'\in C.
\]
Let $B_{N,p}\sim\operatorname{Bin}(N,p)$, with $N$ and $p$ as in
\eqref{eq:packing-parameters}.  Then, for every integer $0\leq h\leq N$,
\begin{equation}
  \frac{|C|}{q^n}
  \leq
  \Pr\{B_{N,p}<h\}
  +\frac{q^{kt}}{\binom{h+t}{t}}.
  \label{eq:finite-length-packing}
\end{equation}
In particular, if $N>0$ and $h=\lfloor pN/2\rfloor$, then
\begin{equation}
  |C|
  \leq
  q^n\exp\!\left(-\frac{p^2N}{2}\right)
  +\frac{q^{n+kt}}{\binom{\lfloor pN/2\rfloor+t}{t}}.
  \label{eq:finite-length-hoeffding}
\end{equation}
\end{theorem}

\begin{proof}
Call a word exceptional when $G_{k,\chi}(x)-1<h$.  By
\eqref{eq:eligible-endpoint-binomial}, the number of exceptional ambient
words is exactly $q^n\Pr\{B_{N,p}<h\}$.  Every remaining codeword has, by
Lemma~\ref{lem:separator-subcone}, at least
\[
  \binom{G_{k,\chi}(x)+t-1}{t}
  \geq
  \binom{h+t}{t}
\]
exact-$t$ descendants.  The descendant spheres of the nonexceptional
codewords are pairwise disjoint, and all lie in $\Sigma^{n+kt}$.  Packing
these spheres and then adding the exceptional codewords proves
\eqref{eq:finite-length-packing}.  Hoeffding's inequality gives
\[
  \Pr\{B_{N,p}<pN/2\}
  \leq
  \exp(-p^2N/2),
\]
which proves \eqref{eq:finite-length-hoeffding}.
\end{proof}

\begin{corollary}[A fixed number of duplications]
\label{cor:fixed-number-packing}
Fix $q\geq2$, $k\geq1$, $t\geq1$, and
$\chi\in\{\pal,\rc\}$.  Every family
$C_n\subseteq\Sigma^n$ with pairwise disjoint exact-$t$ descendant
spheres obeys
\begin{equation}
  |C_n|
  \leq
  \left(
    \frac{t!\,q^{(k+1)t}}{(q-1)^t}+o(1)
  \right)
  \frac{q^n}{n^t}.
  \label{eq:fixed-t-code-size}
\end{equation}
Equivalently, for the $q$-ary redundancy
$r_q(C_n)=n-\log_q|C_n|$,
\begin{equation}
  r_q(C_n)
  \geq
  t\log_q n
  +t\log_q(q-1)
  -\log_q(t!)
  -(k+1)t
  -o(1).
  \label{eq:fixed-t-redundancy}
\end{equation}
In particular, a binary code correcting two length-two palindromic or
reverse-complement duplications satisfies
\[
  |C_n|
  \leq
  (128+o(1))\frac{2^n}{n^2},
  \qquad
  r_2(C_n)\geq2\log_2 n-7-o(1).
\]
\end{corollary}

\begin{proof}
With $N=n-k$ as in \eqref{eq:packing-parameters}, put
\[
  a_n\eqdef\sqrt{\frac{t+1}{2}\,N\ln n},
  \qquad
  h_n\eqdef\lfloor pN-a_n\rfloor.
\]
For all sufficiently large $n$, $h_n\geq0$, and Hoeffding's inequality
gives
\[
  \Pr\{B_{N,p}<h_n\}
  \leq
  \exp\!\left(-\frac{2a_n^2}{N}\right)
  =n^{-(t+1)}.
\]
Moreover, $h_n=pn-o(n)$ and hence
\[
  \binom{h_n+t}{t}
  =
  \left(\frac{p^t}{t!}+o(1)\right)n^t.
\]
Substitution in \eqref{eq:finite-length-packing}, with $t$ fixed, gives
\[
  |C_n|
  \leq
  \left(\frac{t!\,q^{kt}}{p^t}+o(1)\right)\frac{q^n}{n^t}.
\]
Since $p=(q-1)/q$, this is \eqref{eq:fixed-t-code-size}; taking
$q$-ary logarithms gives \eqref{eq:fixed-t-redundancy}.  At
$q=k=t=2$, the constant is $2!\,2^6=128$.
\end{proof}

The constant $t!\,q^{(k+1)t}/(q-1)^t$ in \eqref{eq:fixed-t-code-size} coincides with the constant in the sphere-packing bound of Lenz,
J\"unger, and Wachter-Zeh for tandem duplications
~\cite[Lemma~1]{LenzJungerWachterZeh2018}, with $G_{k,\chi}(x)-1$ in the role
of the Hamming weight of the $(n-k)$-symbol word $z$ of $k$-step differences of $x$ in their map $\phi(x)=(y,z)$.

\begin{remark}[Relation to constant runs]
We compare $G_{k,\pal}$ with a run count in the notation of~\cite{LenzWachterZehYaakobi2019}: let $r^{(\geq k)}(x)$ denote the number
of maximal constant runs of $x$ having length at least $k$.  For the
palindromic channel, distinct such runs end at distinct eligible endpoints, so
$G_{k,\pal}(x)\geq r^{(\geq k)}(x)$ pointwise, and Lemma~\ref{lem:separator-subcone} gives
\[
  |D^t_{k,\pal}(x)|
  \geq
  \binom{r^{(\geq k)}(x)+t-1}{t}.
\]
This domination can be
strict by a linear amount when $k\geq2$: an alternating word has
$r^{(\geq k)}(x)=0$ but $G_{k,\pal}(x)=n-k+1$.  Consequently, for any common lower threshold, at most as many source words have
$G_{k,\pal}(x)$ below it as have $r^{(\geq k)}(x)$ below it, and the endpoint-multiplicity family of every word has at least as many descendants
as its subfamily with copies only after the ends of runs of length at least $k$;
for fixed $q,k,t$, both approaches nevertheless give the same leading
term $t\log_q n$ in the redundancy bound.  Indeed, among
$\lfloor n/(k+1)\rfloor$ disjoint length-$(k+1)$ blocks, the events that a
block has the form $a^k b$ with $a\ne b$ are independent, each with
probability $(q-1)/q^k$, and their successes certify distinct maximal runs of
length at least $k$.  A Chernoff bound therefore gives
$r^{(\geq k)}(x)=\Omega(n)$ outside an $e^{-\Omega(n)}$ fraction of words,
so the corresponding fixed-$t$ subfamilies have size $\Theta(n^t)$.
\end{remark}

\begin{corollary}[A sublinear number of duplications]
\label{cor:sublinear-number-packing}
Fix $q\geq2$, $k\geq1$, and $\chi\in\{\pal,\rc\}$, and let
$1\leq t=t(n)=o(n)$ be integer-valued.  For every fixed $0<\alpha<p$, every
family $C_n\subseteq\Sigma^n$ with pairwise disjoint exact-$t(n)$
descendant spheres obeys, as $n\to\infty$,
\begin{equation}
  |C_n|
  \leq
  (1+o(1))q^n
  \left(
    \frac{q^k t}{\alpha(n-k)}
  \right)^t.
  \label{eq:sublinear-t-code-size}
\end{equation}
Choosing, for example, $\alpha\eqdef p/2$, we consequently obtain
\begin{equation}
  r_q(C_n)
  \geq
  t\log_q(n/t)-O_{q,k}(t).
  \label{eq:sublinear-t-redundancy}
\end{equation}
If in addition $t\to\infty$, then the sharper form
\[
  r_q(C_n)
  \geq
  t\log_q\!\left(\frac{n-k}{t}\right)
  -kt+t\log_q(pe)-o(t)
\]
holds.
\end{corollary}

\begin{proof}
With $N=n-k$ as in \eqref{eq:packing-parameters}, take $h\eqdef\lfloor\alpha N\rfloor$.  For all sufficiently
large $n$, $h\geq1$.  Hoeffding's inequality bounds the first term of
\eqref{eq:finite-length-packing} by
$\exp(-2(p-\alpha)^2N)$.  Also,
\[
  \binom{h+t}{t}
  \geq
  \frac{h^t}{t!},
  \qquad
  t!\leq t^t,
\]
so the second term is at most
\[
  \left(\frac{q^kt}{h}\right)^t
  =
  (1+o(1))
  \left(\frac{q^kt}{\alpha N}\right)^t,
\]
where $(\alpha N/h)^t=\exp(O(t/N))=1+o(1)$.  Because $t=o(n)$, the
logarithm of the ratio of the Hoeffding bound to $(q^kt/(\alpha N))^t$ satisfies
\[
  -2(p-\alpha)^2N
  +t\ln\!\left(\frac{\alpha N}{q^kt}\right)
  =-\Theta(n)+o(n).
\]
The binomial-tail term is therefore negligible relative to this upper
envelope, including when $t$ is bounded or oscillates.  This proves
\eqref{eq:sublinear-t-code-size}; choosing $\alpha\eqdef p/2$, taking logarithms,
and using $t\log_q(n/(n-k))=o(1)$ proves
\eqref{eq:sublinear-t-redundancy}.

For the sharper assertion, assume $t\to\infty$ and put
$L\eqdef\ln(N/t)$,
$a\eqdef\sqrt{NtL}$, and now take $h\eqdef\lfloor pN-a\rfloor$.  Then $a=o(N)$,
the binomial tail is at most $\exp(-2tL)$, and Stirling's formula gives
\[
  \ln\binom{h+t}{t}
  =
  t\ln(N/t)+t\ln p+t+o(t).
\]
The tail is negligible relative to the packing term
$\exp(-tL+O(t))$.  Taking logarithms yields the claimed refinement.
\end{proof}

In particular, these bounds apply to reverse-complement duplication at every
length $k$.  Section~\ref{sec:even-equivalence} develops the additional
channel equivalence available when $k$ is even.

\begin{remark}[Asymptotic scope]
In Corollary~\ref{cor:fixed-number-packing}, $q,k,t$ are fixed while
$n\to\infty$.  In Corollary~\ref{cor:sublinear-number-packing}, only $q,k$
are fixed and $t=o(n)$.  The implied constants depend on $q$ and $k$, and no uniformity is claimed when
the alphabet size or the duplication length grows with $n$.

At $k=1$, the palindromic endpoint-multiplicity family coincides with the
classical sticky-insertion sphere; see, for example,
\cite{MahdavifarVardy2017}.  For a single palindromic duplication,
Lenz, Wachter-Zeh, and Yaakobi give exact collision conditions and exact
sphere sizes for lengths one and two~\cite{LenzWachterZehYaakobi2019}.
\end{remark}

\section{Even-length equivalence of the two channels}
\label{sec:even-equivalence}

When the duplication length is even, a change of coordinates conjugates
each reverse-complement duplication to a palindromic duplication and maps
reverse-complement descendant sets onto palindromic ones.  For
$x=x_0\cdots x_{n-1}\in\Sigma^n$, define the alternating-complement map
$\Amap_n:\Sigma^n\to\Sigma^n$ by
\begin{equation}
  \bigl(\Amap_n(x)\bigr)_j\eqdef \iota^j(x_j),
  \qquad 0\leq j<n.
  \label{eq:Amap}
\end{equation}
Thus $\Amap_n$ complements precisely the odd coordinates.  Because
$\iota^2=\mathrm{id}$, we have $\Amap_n^2=\mathrm{id}$; hence $\Amap_n$ is a
bijection and an involution.

\begin{theorem}[Even-length conjugacy]
\label{thm:conjugacy}
Let $k\geq2$ be even.  For every integer $n\geq k$, every
$i\in\{0,\ldots,n-k\}$, and every $x\in\Sigma^n$,
\begin{equation}
  \Amap_{n+k}\!\left(T^{\rc}_{i,k}(x)\right)
  =T^{\pal}_{i,k}\!\left(\Amap_n(x)\right).
  \label{eq:conjugacy}
\end{equation}
The identity remains valid step by step when the duplication lengths vary,
provided that every applied length is even.
\end{theorem}

\begin{proof}
Write $x=uvw$, where $|u|=i$ and $|v|=k$.  The coordinates of the prefix
$uv$ do not move, so the two sides of
\eqref{eq:conjugacy} agree there.  Fix an offset
$s\in\{0,\ldots,k-1\}$ in the inserted block and put $r\eqdef i+k-1-s$.
The reverse-complement symbol inserted at coordinate $i+k+s$ is
$\iota(x_r)$.  After applying $\Amap_{n+k}$ it becomes
\[
  \iota^{i+k+s}\!\left(\iota(x_r)\right)
  =\iota^{i+k+s+1}(x_r)
  =\iota^r(x_r),
\]
because $(i+k+s+1)-r=2s+2$ is even.  As $s$ increases, $r$ decreases, so
the transformed inserted block is exactly the reverse of the corresponding
factor of $\Amap_n(x)$.

An original suffix symbol at coordinate $j$ moves to coordinate $j+k$.
Since $k$ is even, $\iota^{j+k}=\iota^j$, and its transformed symbol is
unchanged.  This proves the identity.  Applying it after each channel
operation proves the variable-length statement.
\end{proof}

For a set $S\subseteq\Sigma^n$, write
$\Amap_n(S)\eqdef\{\Amap_n(x):x\in S\}$.

\begin{corollary}[Full channel transfer]
\label{cor:transfer}
Let $k\geq2$ be even, let $n\geq k$, and let $x\in\Sigma^n$.  For every
$s\geq0$,
\[
  \Amap_{n+sk}\!\left(D^s_{k,\rc}(x)\right)
  =D^s_{k,\pal}\!\left(\Amap_n(x)\right).
\]
Taking the disjoint union over output lengths gives, for every $t\geq0$,
\begin{equation}
  \bigsqcup_{s=0}^{t}
  \Amap_{n+sk}\!\left(D^s_{k,\rc}(x)\right)
  =
  \bigsqcup_{s=0}^{t}
  D^s_{k,\pal}\!\left(\Amap_n(x)\right).
  \label{eq:at-most-transfer}
\end{equation}
Consequently, $\Amap_n$ is an isomorphism from each reverse-complement exact-$s$
(respectively, at-most-$t$) confusability graph onto the corresponding
palindromic graph.  A code $C\subseteq\Sigma^n$ corrects
exactly $s$ (respectively, at most $t$) palindromic duplications of length $k$ if and only if $\Amap_n(C)$ corrects the
corresponding reverse-complement errors.  Sphere sizes, optimal code cardinalities,
and redundancies agree.  The corresponding encoder and decoder maps are made
explicit in Corollary~\ref{cor:decoder-transfer}.
\end{corollary}

\begin{proof}
Iterate Theorem~\ref{thm:conjugacy} for each fixed value of $s$.  The
$s$-th layer has output length $n+sk$, so distinct layers are disjoint and
require their own map $\Amap_{n+sk}$.  The remaining set-theoretic claims
follow from these layerwise bijections.
\end{proof}

To decode a reverse-complement output using Corollary~\ref{cor:transfer},
first transform the received word to palindromic coordinates.  The following
formulas specify the corresponding codeword and message decoders.

\begin{corollary}[Encoder and decoder transfer]
\label{cor:decoder-transfer}
Let $k\geq2$ be even, let $n\geq k$ and $t\geq0$, and let
$C_{\pal}\subseteq\Sigma^n$ correct at most $t$
length-$k$ palindromic duplications, and put
\[
  C_{\rc}\eqdef\Amap_n(C_{\pal}).
\]
Define the valid received sets
\[
  \mathcal Y_{\pal}
  \eqdef
  \bigsqcup_{c\in C_{\pal}}\bigsqcup_{s=0}^{t}D^s_{k,\pal}(c),
  \qquad
  \mathcal Y_{\rc}
  \eqdef
  \bigsqcup_{x\in C_{\rc}}\bigsqcup_{s=0}^{t}D^s_{k,\rc}(x).
\]
Suppose that $\Dec_{\pal}:\mathcal Y_{\pal}\to C_{\pal}$ is a codeword
decoder satisfying
\[
  \Dec_{\pal}(z)=c
  \quad\text{whenever}\quad
  c\in C_{\pal},\qquad
  z\in D^s_{k,\pal}(c),\qquad 0\leq s\leq t.
\]
Then a codeword decoder
$\Dec_{\rc}:\mathcal Y_{\rc}\to C_{\rc}$ is
\begin{equation}
  \Dec_{\rc}(y)
  \eqdef
  \Amap_n\!\left(
    \Dec_{\pal}\!\left(\Amap_{|y|}(y)\right)
  \right).
  \label{eq:decoder-transfer}
\end{equation}
If $\mathcal M$ is a message set and
$E_{\pal}:\mathcal M\to C_{\pal}$ is an encoder, then
$E_{\rc}\eqdef\Amap_n\circ E_{\pal}$ encodes the same message set into
$C_{\rc}$.  If
$\widehat{\Dec}_{\pal}:\mathcal Y_{\pal}\to\mathcal M$ returns the encoded
message on every valid received word, its reverse-complement version is
simply the map $\widehat{\Dec}_{\rc}:\mathcal Y_{\rc}\to\mathcal M$ given by
\[
  \widehat{\Dec}_{\rc}(y)
  \eqdef\widehat{\Dec}_{\pal}\!\left(\Amap_{|y|}(y)\right).
\]
Both transforms are single coordinate-wise passes.  Thus they add
$O(n)$ time to the encoder and $O(|y|)$ time to the decoder.  Under a mutable
random-access word representation and a constant-time implementation of
$\iota$, either transform may be performed in place with $O(1)$ additional
working memory; an out-of-place implementation uses a linear output buffer.
\end{corollary}

\begin{proof}
Let $x\in C_{\rc}$ and $y\in D^s_{k,\rc}(x)$ for some $0\leq s\leq t$.
Since $\Amap_n$ is an involution,
$c\eqdef\Amap_n(x)$ belongs to $C_{\pal}$.  Corollary~\ref{cor:transfer}
gives
\[
  \Amap_{|y|}(y)=\Amap_{n+sk}(y)
  \in D^s_{k,\pal}(c).
\]
Hence $\Dec_{\pal}(\Amap_{|y|}(y))=c$, and the outer transform in
\eqref{eq:decoder-transfer} returns $\Amap_n(c)=x$.  The encoder and
message-decoder formulas follow from the same identities.  The complexity
claims follow because $\Amap_m$ visits each of the $m$ coordinates once.
\end{proof}

\begin{remark}[Histories with variable even lengths]
Theorem~\ref{thm:conjugacy} also covers histories whose duplication lengths are all even;
this extends the transfers to every channel defined by a source-independent
family of allowed histories.
A finite history is a sequence
$h=((i_1,k_1),\ldots,(i_s,k_s))$, where $s\geq0$ and every $k_r$ is even.
For $s=0$, this is the empty history.  The history is
legal from $x\in\Sigma^n$ when, at step $r$, the start $i_r$ lies in the
legal range for the word obtained after the preceding $r-1$ steps.  For a
source-independent family $\mathcal H$ of such formal sequences, let
$\mathcal H(x)$ be its histories that are legal from $x$.  Write
$T_h^\chi(x)$ for the descendant obtained by sequentially applying the steps of
$h$ in channel $\chi$.  For $s\geq1$, the order is
$T_h^\chi(x)=T^\chi_{i_s,k_s}(\cdots T^\chi_{i_1,k_1}(x)\cdots)$; for the
empty history, $T_h^\chi(x)=x$.
Legality depends only on the evolving word length, so
$\mathcal H(x)=\mathcal H(\Amap_n(x))$.  Writing
$L(h)\eqdef\sum_{r=1}^s k_r$, with the empty sum equal to zero, repeated use of
Theorem~\ref{thm:conjugacy} gives,
for every total inserted length $L$, the layerwise bijection
\[
  \Amap_{n+L}\!\left(
    \{T_h^{\rc}(x):h\in\mathcal H(x),\ L(h)=L\}
  \right)
  =
  \{T_h^{\pal}(\Amap_n(x)):h\in\mathcal H(x),\ L(h)=L\}.
\]
Taking the union over $L$ gives the corresponding descendant-set and
confusability transfers.  The decoder transfer of Corollary~\ref{cor:decoder-transfer} applies in the same way,
with $\Amap_n$ on the source and $\Amap_{|y|}$ on the received word.
\end{remark}

\begin{remark}[Odd-length suffix parity]
Only the suffix calculation in the proof of Theorem~\ref{thm:conjugacy} uses that $k$ is
even: when $k$ is odd, every suffix coordinate
changes parity.  Thus this map does not in general conjugate the full
channels at odd length.  Parity partitions appear in the
expressiveness analysis of Ben-Tolila and Schwartz
~\cite{BenTolilaSchwartz2022}; their set of the symbols of a word $s\in\Sigma^n$ at even positions together with the complements of
the symbols at odd positions is exactly the set of symbols of $\Amap_n(s)$, and
for even $k$ their proof of~\cite[Theorem~3]{BenTolilaSchwartz2022} uses that
original symbols keep the parity of their positions, but it does not give the bijective
even-length conjugacy above.
\end{remark}

The next consequence does not need the conjugacy.
As a general benchmark for both channels and every $k$, the arbitrary-burst
construction of Sun, Lu, Zhang, and Ge gives the following.

\begin{corollary}[One fixed-length duplication]
\label{cor:one-fixed-length}
Fix $q\geq2$ and $k\geq1$.  There are $q$-ary codes correcting one
palindromic duplication of length $k$, with
\[
  r_q(C_n)\leq\log_q n+O_{q,k}(1),
\]
and likewise for one reverse-complement duplication.  For $k=1$ the code
of~\cite{SunLuZhangGe2025} is the explicit Tenengolts construction.  For
$k\geq2$, the construction there is a syndrome-parameterized ensemble containing such a code;
once a suitable parameter tuple is supplied, its decoder is constructive,
whereas the size guarantee is existential over the tuple.
The displayed redundancy is optimal up to
an additive constant.
\end{corollary}

\begin{proof}
Sun, Lu, Zhang, and Ge work over the $q$-ary alphabet $\Sigma_q$.  For
$k\geq2$, their Theorem~9 explicitly describes a syndrome-parameterized
ensemble of $q$-ary codes for one burst of $k$ deletions and proves by
pigeonhole that a parameter tuple attaining the stated size exists.  Once
that tuple is supplied, the decoder is constructive.  The deletion--insertion
equivalence recalled there applies to the same length-$n$ syndrome class,
which therefore corrects one burst of $k$
insertions~\cite[Theorem~9]{SunLuZhangGe2025}; testing membership after each of the
$n+1$ possible deletions of a length-$k$ factor then decodes such a burst constructively.  The
case $k=1$ is supplied by the classical $q$-ary Tenengolts construction
recalled immediately before Definition~3 of that paper
~\cite{SunLuZhangGe2025}.  The selected codes have
$\log_2 n+O_{q,k}(1)$ bits of redundancy.  A
palindromic or reverse-complement duplication inserts one particular
length-$k$ block, so both duplication channels are subchannels of that burst
channel.  After division by $\log_2 q$, the redundancy is
$\log_q n+O_{q,k}(1)$ in $q$-ary symbols.  The matching lower bound follows
from Corollary~\ref{cor:fixed-number-packing} with $t=1$; for RC codes with odd $k$ and a
fixed-point-free complement, \cite[Theorem~18]{BenTolilaSchwartz2022} gives
the slightly stronger bound $|C|\leq q^{n+k}/(n-k+1)$.
\end{proof}

\begin{remark}[Comparison with the odd-length one-error result]
Ben-Tolila and Schwartz constructed codes for one reverse-complement
duplication of fixed odd length~\cite{BenTolilaSchwartz2022}.  For $k=1$
their redundancy is $\log_2 n+o(1)$ bits.  For every fixed odd
$k\geq3$, their displayed burst-code instantiation has redundancy at most
\[
  \log_2 n+(k-1)\log_2\!\log_2 n+O_k(1),
\]
and this bound is optimal in the leading coefficient but not up to an
additive constant.  Corollary~\ref{cor:one-fixed-length}, obtained from the
newer arbitrary-burst construction, is optimal up to an additive $O_{q,k}(1)$
term for one fixed-length duplication of either parity.
\end{remark}

The conjugacy can be combined with any palindromic construction.  The next
application transfers a known channel-specific palindromic code without a
new decoding argument; it shows the conjugacy acting on an explicit
run-length syndrome, while generic arbitrary-burst codes already give
redundancy $\log_2 n+O_k(1)$ (Corollary~\ref{cor:one-fixed-length}).  For a binary word $z$, let $r(z)$ be its
number of maximal constant runs, let $r_j(z)$ be the length of its $j$th
run from the left, and let $r^{(1)}(z)$ be its number of singleton runs.
For $x\in\{0,1\}^n$, define the pulled-back run statistics
\[
  r_{\rm alt}^{(1)}(x)\eqdef r^{(1)}(\Amap_n(x)),
  \qquad
  \mathrm{cs}_{\rm alt}(x)
  \eqdef
  \sum_{j=1}^{r(\Amap_n(x))}j\,r_j(\Amap_n(x)).
\]
The maximal alternating factors of $x$ are precisely the constant runs of
$\Amap_n(x)$.  For $a\in\Z_5$ and $b\in\Z_{2n+1}$, put
\[
  \mathcal C^{\rc}_{a,b}(n)
  \eqdef
  \left\{x\in\{0,1\}^n:
  r_{\rm alt}^{(1)}(x)\equiv a\pmod 5,
  \mathrm{cs}_{\rm alt}(x)\equiv b\pmod{2n+1}\right\}.
\]

\begin{corollary}[A transferred one-error code]
\label{cor:one-error}
Each $\mathcal C^{\rc}_{a,b}(n)$ corrects one binary reverse-complement
duplication of length two.  Moreover, at least one pair $(a,b)$ satisfies
\[
  |\mathcal C^{\rc}_{a,b}(n)|\geq \frac{2^n}{5(2n+1)}.
\]
Its redundancy is at most
$\log_2(5(2n+1))=\log_2 n+\log_2 10+o(1)$.
\end{corollary}

\begin{proof}
Lenz, Wachter-Zeh, and Yaakobi prove the corresponding statement for one
binary palindromic duplication of length two in their Definitions~8--9,
Construction~2, and Theorem~2, using runs numbered from left to right
starting at one~\cite{LenzWachterZehYaakobi2019}.  The statistics
$r_{\rm alt}^{(1)}(x)$ and $\mathrm{cs}_{\rm alt}(x)$ are exactly the ordinary run
statistics of $\Amap_n(x)$, so $\mathcal C^{\rc}_{a,b}(n)$ is the inverse
image under $\Amap_n$ of their syndrome class.  Corollary~\ref{cor:transfer}
transfers the correction property.  The $5(2n+1)$ classes partition
$\{0,1\}^n$, and
averaging gives the cardinality bound.
\end{proof}

\section{Two-error sphere geometry at even lengths}
\label{sec:edge-tree}

We first determine the maximum exact-two descendant sphere for every
even duplication length.  Section~\ref{sec:existence} uses this maximum to bound the degree
of the two-error confusability graph.  We then return to length two, where a finer
edge-tree normal form assigns each two-step history a structural shape that
determines its descendant.

\subsection{The exact maximum for every even length}

Throughout this subsection, let $|\Sigma|\geq2$, let $k\geq2$ be even,
and let $n\geq k$.
For $x\in\Sigma^n$, call $(i,j)$ a two-step history when
$0\leq i\leq n-k$ is the start of the first duplication and
$0\leq j\leq n$ is the start of the second duplication in the resulting
length-$(n+k)$ word.  Write
\[
  \Phi_x(i,j)
  \eqdef
  T^{\pal}_{j,k}\!\left(T^{\pal}_{i,k}(x)\right).
\]

\begin{theorem}[Maximum exact-two sphere at even duplication length]
\label{thm:even-k-two-sphere}
Let $|\Sigma|\geq2$, let $k\geq2$ be even, and let $n\geq k$.  Put
$a_+\eqdef\max\{a,0\}$, with the convention
$\binom{0}{2}=\binom{1}{2}=0$, and define
\[
  S^{\max}_{k,2}(n)
  \eqdef
  \max_{x\in\Sigma^n}\bigl|D^2_{k,\pal}(x)\bigr|.
\]
Then
\begin{equation}
  S^{\max}_{k,2}(n)
  =(n-k+1)(n+1)-\binom{(n-2k+2)_+}{2}.
  \label{eq:even-k-two-sphere}
\end{equation}
Equivalently, the right-hand side is
\[
  \binom{n-k+2}{2}
  +\sum_{i=0}^{n-k}\min\{2k-1,n-i\}.
\]
For $n\geq2k-2$, it simplifies to
\begin{equation}
  \frac{n^2+(2k+1)n-4k(k-1)}{2}.
  \label{eq:even-k-two-sphere-stable}
\end{equation}
Every alternating word on any two distinct symbols of $\Sigma$ attains the
maximum.
\end{theorem}

The two orders of duplicating disjoint source factors give the same
descendant for every source.  To show that these are the only collisions
on an alternating source, we use Theorem~\ref{thm:conjugacy} on its two
symbols.  Denote the initial symbol by $0$ and the other by $1$, and use
the involution exchanging them on this binary alphabet, which we may do because
palindromic duplication applies no complement and every descendant of the alternating source is a
word over these two symbols.  With this choice,
$\Amap_n$ sends the alternating source $u$ to $0^n$ and preserves both
duplication starts when passing from PAL to RC.

\begin{lemma}[Run lengths after two duplications]
\label{lem:alternating-six-regime-signatures}
For a history $(i,j)$, put $d\eqdef j-i$ and
$v\eqdef\Amap_{n+2k}(\Phi_u(i,j))$.  The first $1$ in $v$ has
coordinate $i+k+\min\{d,0\}$.  Removing all leading and trailing zeros
from $v$ leaves the word in the following table.
\[
\begin{array}{c|c@{\qquad}c|c}
d<0 & 1^k0^{-d}1^k & d=0 & 1^{2k}\\
0<d<k & 1^d0^d1^{2(k-d)} & d=k & 1^k\\
k<d<2k & 1^k0^{d-k}1^{d-k} & d\geq2k & 1^k0^{d-k}1^k
\end{array}
\]
\end{lemma}

\begin{proof}
The first RC duplication gives $0^{i+k}1^k0^{n-i-k}$.
The second inserts its reverse-complement block at $i+k+d$.
For $d\leq0$ and $d\geq2k$ its source is $0^k$, so it inserts $1^k$.
At $d=0$ the two runs join.  For $0<d<k$ the source is
$0^{k-d}1^d$, and the inserted block $0^d1^{k-d}$ precedes the
remaining $1^{k-d}$.  For $d=k$ the source is $1^k$ and the
inserted block is $0^k$.  For $k<d<2k$ the source is
$1^{2k-d}0^{d-k}$ and the inserted block is $1^{d-k}0^{2k-d}$.
These descriptions also give the coordinate of the first $1$.
\end{proof}

For example, with $k=4$, $n=10$, and $(i,j)=(1,3)$, the transformed
output $v$ is $0^5 1^2 0^2 1^4 0^5$.  Its first run of two ones gives
$d=2$, and its five initial zeros give $i=1$.

\begin{lemma}[Collision classification on an alternating source]
\label{lem:alternating-collision-classification}
For two-step histories $(i,j)$ and $(i',j')$ on the alternating source $u$,
\[
  \Phi_u(i,j)=\Phi_u(i',j')
\]
if and only if the histories are identical or, after interchanging them, there
are integers $r,d$ with $d\leq-k$ such that
\begin{equation}
  (i,j)=(r,r+d),
  \qquad
  (i',j')=(r+d,r+k).
  \label{eq:alternating-commuting-pair-form}
\end{equation}
\end{lemma}

\begin{proof}
The run lengths and the position of the first $1$ in
Lemma~\ref{lem:alternating-six-regime-signatures} determine the entire
transformed output.  A single run of ones has length $2k$ or $k$, distinguishing
$d=0$ from $d=k$.  With two runs of ones, a first run shorter than $k$ has length
$d$ and identifies $0<d<k$.  A first run of length $k$ followed by a
shorter run identifies $k<d<2k$ and gives $d-k$ as the second length.
In these cases the coordinate of the first $1$ recovers $i$.

It remains to consider two runs of length $k$ separated by $g$ zeros.
A history with $d<0$ has $d=-g$, and a history with $d'\geq2k$ has
$d'=k+g$.  Each determines the starts within its own case from the
coordinate of the first $1$.  Both cases can occur for the same output
only when $g\geq k$.  Equality of the coordinates of the first $1$ then gives
$i'=i+d$ and hence $j'=i+k$.  Conversely, the two histories in
\eqref{eq:alternating-commuting-pair-form} give these same runs at the
same coordinates.  The binary conjugacy transfers equality back to PAL.
\end{proof}

\begin{proof}[Proof of Theorem~\ref{thm:even-k-two-sphere}]
There are $H\eqdef(n-k+1)(n+1)$ ordered two-step histories.  Suppose that
$0\leq a<b\leq n-k$ and $b-a\geq k$.  The two length-$k$ source factors
starting at $a$ and $b$ are disjoint, and their duplications commute:
\begin{equation}
  \Phi_x(a,b+k)=\Phi_x(b,a).
  \label{eq:commuting-histories}
\end{equation}
Indeed, if $x=pvzws$, where $|p|=a$, $|v|=|w|=k$,
$|z|=b-a-k$, and $|s|=n-b-k$, then both sides are
$pvv^Rzww^Rs$.  These forced pairs of histories are pairwise disjoint.  In
the left member of such a pair, $j-i\geq2k$ and the original starts are
$(i,j-k)=(a,b)$; in the right member, $j-i\leq-k$ and the original starts
are $(j,i)=(a,b)$.  Their number is
\[
  P
  \eqdef\#\{(a,b):0\leq a<b\leq n-k,\ b-a\geq k\}
  =\binom{(n-2k+2)_+}{2}.
\]
Consequently, every exact-two descendant sphere has size at most $H-P$.

For attainability, choose two distinct symbols of $\Sigma$, denote them by
$0,1$, and take the alternating source $u=u_0\cdots u_{n-1}$ used above.
Lemma~\ref{lem:alternating-collision-classification} says that its history
fibers are singletons except for the pairs
\eqref{eq:alternating-commuting-pair-form}.  Writing $a\eqdef r+d$ and $b\eqdef r$
identifies these pairs bijectively with the $P$ choices in
\eqref{eq:commuting-histories}.  Thus there are exactly $P$ doubleton fibers
and no larger fibers, so
$|D^2_{k,\pal}(u)|=H-P$.  This proves
\eqref{eq:even-k-two-sphere} and the equality claim.

The $P$ histories with $j\geq i+2k$ are exactly the left members of the forced
pairs, so $H-P$ counts the histories with $j<i+2k$: for each $i$, the $i+1$ starts
$j\leq i$ and the $\min\{2k-1,n-i\}$ starts $i<j<i+2k$.  This is the equivalent form.  When $n\geq2k-2$, expanding $H-P$ gives
\eqref{eq:even-k-two-sphere-stable}.
\end{proof}

The following corollary shows that the one-error sphere of the same
alternating source has the maximum possible size $n-k+1$.  Since the zero-,
one-, and two-error outputs have different lengths, this source also maximizes the
at-most-two ball; Corollary~\ref{cor:transfer} gives the reverse-complement statements.

\begin{corollary}[At-most-two balls and reverse-complement transfer]
\label{cor:rc-sphere}
For every even $k\geq2$ and $n\geq k$,
\[
  \max_{x\in\Sigma^n}\bigl|D^{\leq2}_{k,\pal}(x)\bigr|
  =1+(n-k+1)+(n-k+1)(n+1)
   -\binom{(n-2k+2)_+}{2}.
\]
The same exact-two and at-most-two maxima hold for reverse-complement
duplication.  In the binary reverse-complement channel they are attained by the constant
sources $0^n$ and $1^n$.
\end{corollary}

\begin{proof}
A one-error sphere has at most $n-k+1$ elements, since one operation is
determined by its start.  After applying $\Amap_{n+k}$, as in the proof of
Lemma~\ref{lem:alternating-six-regime-signatures}, the one-error output
of $u$ is $0^{i+k}1^k0^{n-i-k}$, whose run $1^k$ begins at $i+k$ and identifies $i$.
Thus all these descendants are distinct.  This source simultaneously
attains the two-error maximum, proving the palindromic claim.  The
reverse-complement claim follows from Corollary~\ref{cor:transfer}.  In the
binary case,
$\Amap_n(0^n)$ and $\Amap_n(1^n)$ are alternating.
\end{proof}

For illustration, formula~\eqref{eq:even-k-two-sphere-stable} gives
\begin{align*}
 k=4:&\qquad
 \max_x|D^2_{4,\pal}(x)|=\frac{n^2+9n-48}{2},
 && n\geq6,\\
 k=6:&\qquad
 \max_x|D^2_{6,\pal}(x)|=\frac{n^2+13n-120}{2},
 && n\geq10.
\end{align*}
These formulas
hold verbatim for the corresponding reverse-complement channels.

\begin{remark}[Maximum spheres and packing]
The maximum sphere grows as $\Theta(n^2)$, but descendant spheres are highly
nonuniform.  Thus the maximum-sphere formula alone does not imply a
$2\log_2 n-O_k(1)$ redundancy lower bound.  Lemma~\ref{lem:separator-subcone} and the second bound in
Theorem~\ref{thm:finite-length-packing} instead give quadratic subfamilies for
all but an exponentially small fraction of binary sources, and, for every
length $k$, the fixed-error converse in Corollary~\ref{cor:fixed-number-packing}
gives, for binary codes correcting two length-$k$ palindromic or
reverse-complement duplications,
\[
  r_2(C_n)\geq2\log_2 n-(2k+3)-o(1).
\]
\end{remark}

\subsection{The length-two edge-tree refinement}

We now specialize to the binary palindromic duplication channel with
duplication length two.  We abbreviate binary palindromic duplication of
length two as PAL2 and binary reverse-complement duplication of length two
as RC2.  Throughout this subsection, assume $n\geq2$.  The exact one-error
PAL2 sphere was determined
in~\cite{LenzWachterZehYaakobi2019}; here the edge labels retain additional
ancestry information about two unrestricted sequential errors.  No lemma, theorem
or corollary in Sections~\ref{sec:existence} and~\ref{sec:target} depends on this
subsection.  Its structural shapes record which two-step histories always give the
same descendant, which matters for code constructions: a decoder must separate the
descendants of different codewords, not the histories that produce them.
Appendix~\ref{app:construction-obstructions} uses the shapes in an obstruction for
additive weights on individual error histories, and
Remark~\ref{rem:duplication-specific-route} uses them to suggest labeling classes of
histories instead.

For a binary word $z$ and $0\leq h\leq |z|-2$, an indexed adjacency
occurrence of $z$ is the object
\[
  e_h(z)\eqdef\bigl(z,h;(z_h,z_{h+1})\bigr).
\]
Thus occurrences at different starts remain distinct even when they carry
the same ordered pair of bits.

Fix a source $u\in\{0,1\}^n$.  We use the formal adjacency-label set
\[
  \cU_n\eqdef
  \{P_i:0\leq i\leq n-2\}
  \sqcup
  \{L_i,M_i:0\leq i\leq n-2\}
  \sqcup
  \{R_i:0\leq i\leq n-3\}.
\]
The label $P_i$ is assigned to the original occurrence
$e_i(u)$.  If this occurrence survives a duplication, it retains its
label $P_i$ even when its current start shifts.

Let $u^{(i)}\eqdef T^{\pal}_{i,2}(u)$, and write
$a\eqdef u_i$ and $b\eqdef u_{i+1}$.  Locally, $T^{\pal}_{i,2}$ replaces $ab$ by $abba$.
The two new internal adjacency occurrences receive the labels
\[
  L_i\longleftrightarrow
  e_{i+1}(u^{(i)})
  =\bigl(u^{(i)},i+1;(b,b)\bigr),
  \qquad
  M_i\longleftrightarrow
  e_{i+2}(u^{(i)})
  =\bigl(u^{(i)},i+2;(b,a)\bigr).
\]
When $0\leq i\leq n-3$, put $c\eqdef u_{i+2}$ and assign
\[
  R_i\longleftrightarrow
  e_{i+3}(u^{(i)})
  =\bigl(u^{(i)},i+3;(a,c)\bigr).
\]
These labels record the ancestry of adjacency occurrences: $L_i$, $M_i$, and
$R_i$ are the children of $P_i$.  We call this parent--child
structure on adjacency labels the edge tree, and the multiset of the labels of
the two adjacency occurrences duplicated by a two-step history its structural shape.

\begin{lemma}[Two-step edge-tree normal form]
\label{lem:normal-form}
Every ordered history of two PAL2 duplications from a binary word of length
$n$ has one of the following structural shapes:
\begin{enumerate}
  \item a multiset $\{P_i,P_j\}$ of two of the $n-1$ original adjacency
  labels, with repetition allowed;
  \item a parent--child pair $\{P_i,L_i\}$, $\{P_i,M_i\}$, or
  $\{P_i,R_i\}$, subject to the ranges above.
\end{enumerate}
The descendant produced by the history is determined by its shape.  The
number of structural shapes is
\begin{equation}
  \binom n2+2(n-1)+(n-2)
  =\frac{n^2+5n-8}{2},
  \label{eq:two-step-shape-count}
\end{equation}
which equals $S_{2,2}^{\max}(n)$ by Theorem~\ref{thm:even-k-two-sphere}.
\end{lemma}

\begin{proof}
After the first operation $T^{\pal}_{i,2}$, let
$h\in\{0,\ldots,n\}$ be the start of the second operation.  If $h\leq i$,
the second occurrence has the surviving original label $P_h$; the starts
$h=i+1,i+2,i+3$ have labels $L_i,M_i,R_i$, respectively, whenever they
exist; and if $h\geq i+4$, the second occurrence has the surviving original
label $P_{h-2}$.  These cases exhaust all current adjacencies.

The original label $P_{i+1}$ is destroyed when $P_i$ is used first, but the
shape $\{P_i,P_{i+1}\}$ is realized by applying the right original
occurrence first.  More generally, a multiset of two original labels has a
canonical realization: process distinct labels from right to left, or reuse
the same surviving occurrence in the repeated-label case.  Each operation
inserts its reversed pair immediately after the labeled original
occurrence, so every realization of the multiset, this one included, gives
the same descendant.  A parent--child label fixes the first original occurrence and
the second current occurrence, so it also fixes the descendant.  Thus every
ordered history has one of the asserted shapes and each shape determines a
single descendant.

There are $\binom n2$ size-two multisets drawn from $n-1$ original labels,
$2(n-1)$ children of types $L,M$, and $n-2$ children of type $R$.
\end{proof}

Thus the number of structural shapes is a sharp upper bound on the size of
every exact-two PAL2 sphere: different shapes remain distinct on every alternating source, while
other sources may merge them.  The edge-tree description retains ancestry
information that is not visible in the sphere cardinality alone.

\section{Inverse-parent bounds and a direct even-length existence bound}
\label{sec:existence}

To obtain two-error codes by coloring, we bound how many sources can share a
descendant.  Fix an even $k\geq2$; Theorem~\ref{thm:even-k-two-sphere} gives the exact sphere
maximum for these lengths, and the finite-length bound below uses it.  Write PAL-$k$ and RC-$k$ for duplications of
length $k$ in the two channels.  An inverse PAL-$k$ step, a palindromic
deletion in the terminology of~\cite{LenzWachterZehYaakobi2019}, selects a factor
$vv^R$ and deletes the final $v^R$, leaving the source factor $v$.
Enumerating these inverse steps gives the inverse-parent bound below, which remains valid
when the two forward duplications overlap or are nested.  For one duplication of arbitrary
length, Sun and Ge use a similar count~\cite[Lemma~IV.3]{SunGe2026}.

\begin{lemma}[Inverse-parent bound]
\label{lem:inverse-parent}
Let $n\geq k$.  A $q$-ary word $y$ of length $n+k$ has at most $n-k+1$
distinct length-$n$ PAL-$k$ parents at distance one.  A $q$-ary word $y$ of
length $n+2k$ has at most $(n+1)(n-k+1)$ distinct length-$n$ PAL-$k$
ancestors at distance two.
\end{lemma}

\begin{proof}
A word of length $n+k$ has $n-k+1$ length-$2k$ windows.  Every valid last
inverse step must choose a window $vv^R$ and delete its final $k$ symbols,
so there are at most $n-k+1$ one-error parents.

For an output of length $n+2k$, first undo the chronologically last
duplication.  There are at most $n+1$ length-$2k$ windows.  The intermediate
word has length $n+k$, and hence at most $n-k+1$ choices for the second inverse
step.  Every valid two-error history has such a reverse-chronological
description, including histories in which the second duplication acts on or
next to symbols inserted by the first.  Thus the number of ordered inverse
traces, and therefore the number of distinct ancestors, is at most
$(n+1)(n-k+1)$.
\end{proof}

We now work in the exact-two PAL-$k$ confusability graph of
Section~\ref{sec:model}: the number of descendants of a source
(Theorem~\ref{thm:even-k-two-sphere}) times the number of other ancestors
of a descendant (Lemma~\ref{lem:inverse-parent}) bounds the vertex degree.

\begin{theorem}[Direct greedy existence bound at even length]
\label{thm:greedy-code}
Let $q\geq2$, let $k\geq2$ be even, and let $n\geq k$.  Put
$N_1\eqdef n-k+1$, $H_2\eqdef(n+1)N_1$, and
$S_2\eqdef S^{\max}_{k,2}(n)$.  There exists a $q$-ary code correcting at most two
PAL-$k$ duplications with cardinality at least
\begin{equation}
  \frac{q^n}{1+S_2(H_2-1)}.
  \label{eq:greedy-size}
\end{equation}
Consequently, its redundancy is at most
\begin{equation}
  \log_q\!\left(1+S_2(H_2-1)\right)
  =4\log_q n+O_{q,k}(1).
  \label{eq:greedy-redundancy}
\end{equation}
The same statement holds for at most two reverse-complement duplications of
length $k$.
\end{theorem}

\begin{proof}
Work in the exact-two confusability graph.  By
Lemma~\ref{lem:exact-at-most}, every independent set in this graph also
corrects at most two duplications.  A vertex has at most $S_2$ two-error
descendants, and each has at most $H_2-1$ other ancestors at distance two by
Lemma~\ref{lem:inverse-parent}.  Write $\Delta_n$ for the maximum degree.
Then
\[
  \Delta_n\leq S_2(H_2-1).
\]
A greedy $(\Delta_n+1)$-coloring has a color class of size at least
$q^n/(\Delta_n+1)$.  Such a class is an independent set and therefore a
code correcting at most two errors.  Equation~\eqref{eq:greedy-redundancy}
follows from $S_2=\Theta(n^2)$ for fixed $k$, and the reverse-complement
claim follows from Corollary~\ref{cor:transfer}.
\end{proof}

Theorem~\ref{thm:greedy-code} is an existence statement; it does not by
itself provide a structured encoder or decoder.  

\begin{remark}[Comparison with the two-burst existence bound]
\label{rem:two-burst-existence}
Ye et al.\ prove that for every $k\geq2$ and $n>2k$ some $q$-ary code of
length $n$
correcting two bursts of exactly $k$ deletions has at least
$q^{n-2k}/\binom n2^2$ codewords~\cite[Theorem~III.1]{YeSunYuGeElishco2026}.
Wang et al.\ state the equivalence of correcting two bursts of $k$ deletions
and correcting two bursts of $k$ insertions~\cite[Sec.~II]{WangKongYaakobiDuman2026}.
To keep this transfer self-contained,
Lemma~\ref{lem:burst-deletion-to-insertion} below proves the direction needed
here.  Since each PAL or RC duplication inserts a length-$k$ block, these
codes correct at most two PAL-$k$ or RC-$k$ duplications with redundancy at most
\[
  2k+2\log_q\binom n2=4\log_q n+2k-2\log_q2+o(1).
\]
For fixed even $k$, the bound \eqref{eq:greedy-redundancy} equals
$4\log_q n-\log_q2+o(1)$, because $S_2=n^2/2+O(n)$ and $H_2=n^2+O(n)$.  The two
bounds therefore have the same leading term, and the bound of
Theorem~\ref{thm:greedy-code} is smaller by $2k-\log_q2+o(1)$ symbols.  The bound of Theorem~\ref{thm:greedy-code} is also strictly smaller at every
finite length $n\geq2k+1$: there
$H_2\leq(n-1)(n+1)<n^2$, so
$1+S_2(H_2-1)\leq H_2^2<n^4\leq q^{2k}\binom n2^2$.
\end{remark}

Wang et al.\ also give upper bounds and, for fixed $n,t,k$ as $q\to\infty$, a lower bound on the
size of codes correcting multiple bursts of deletions, and state that, for two
bursts of constant length, the best known construction has redundancy at most
$5\log n+o(\log n)$~\cite{WangKongYaakobiDuman2026}.  Liu and Duman prove the analogous
deletion--insertion symmetry for two mixed bursts in binary sequences, but only with nonzero deletion and
insertion lengths~\cite[Theorem~1]{LiuDuman2026}.  We next compare with the syndrome-parameterized
burst-code benchmark of Ye et al.~\cite{YeSunYuGeElishco2026}.

\section{Burst-code benchmarks and open problems}
\label{sec:target}

We compare the redundancy and algorithmic guarantees of two-error duplication
codes with those of arbitrary-burst codes, then formulate the remaining
binary construction problems.  Table~\ref{tab:related-work} places these
comparisons in their respective channel, parity, alphabet, and error-count
regimes.

An arbitrary length-$k$ insertion places a word of $\Sigma^k$ at any gap
of the current word.  The next lemma describes the resulting words after
two insertions, as in the deletion observation
of~\cite[Observation~II.1]{YeSunYuGeElishco2026}.  The disjoint factors
in its statement may be adjacent.

\begin{lemma}[Sequential insertions as disjoint final bursts]
\label{lem:sequential-burst-containment}
Let $x\in\Sigma^n$.  If $y\in\Sigma^{n+2k}$ is obtained from $x$ by two
sequential arbitrary insertions of length $k$, then $y$ contains two
disjoint length-$k$ factors whose simultaneous deletion returns $x$.
Consequently, this containment holds for every two-step PAL-$k$ or RC-$k$
history, even when the second duplicated factor overlaps or is nested in
material created by the first duplication.
\end{lemma}

\begin{proof}
If the second block is inserted outside the interior of the first, both
blocks occupy disjoint intervals in $y$.  Otherwise the inserted material
forms a single interval of length $2k$ in $y$.  Deleting its first and last
$k$ symbols returns $x$.  This argument also applies when the second
inserted block depends on symbols inserted at the first step.
\end{proof}

\begin{lemma}[From two deletion bursts to two insertion bursts]
\label{lem:burst-deletion-to-insertion}
Let $k\geq1$ and $n\geq2k$.  If a code in $\Sigma^n$ corrects two disjoint
bursts of exactly $k$ deletions, it corrects at most two sequential bursts
of exactly $k$ insertions.
\end{lemma}

\begin{proof}
Write $I$ and $D$ for inserting and deleting a length-$k$ block.
Consider a path $a\xrightarrow{I}y\xrightarrow{D}b$ with $|y|\geq2k$.
Mark in $y$ the two intervals whose deletions give $a$ and $b$.
If they are disjoint, deleting both gives a common one-burst deletion
of $a$ and $b$.  If they meet, choose a length-$2k$ interval containing
their union.  After either marked deletion, the remainder of this interval
is a consecutive length-$k$ block.  Deleting it again gives the same word
from $a$ and $b$.  Thus each such $ID$ can be replaced by $DI$.

Two codewords with a common two-insertion output are joined by a path
$IIDD$.  Moving each $D$ past the $I$'s to its left gives $DDII$.
Every exchanged pair has an intermediate word of length at least
$n\geq2k$.  The resulting common two-deletion output comes from two
disjoint bursts on each codeword by Lemma~\ref{lem:sequential-burst-containment}
applied in reverse.  Hence the codewords coincide.  A common one-insertion
output could be extended to a common two-insertion output by appending the
same $k$ symbols, and the received length determines the number of insertions.
\end{proof}

\begin{corollary}[A benchmark from a syndrome-parameterized ensemble for two insertion bursts]
\label{cor:explicit-two-burst-benchmark}
For every fixed $k\geq2$ and every $n\geq\max\{9,2k+1\}$, the
syndrome-parameterized ensemble of Ye et al.\
\cite[Theorem~IV.1]{YeSunYuGeElishco2026}, with the row windows specified
below, contains a binary code $C_n$ correcting at most two sequential arbitrary
length-$k$ insertions, hence at most two unrestricted sequential PAL-$k$ or
RC-$k$ duplications, with
\begin{equation}
  r_2(C_n)\leq 5\log_2 n+14k\log_2 \log_2 n+O_k(1).
  \label{eq:explicit-two-burst-benchmark}
\end{equation}
For a supplied suitable syndrome tuple, the code has a polynomial-time
codeword decoder.  This is a nonuniform guarantee; efficient parameter
selection and message encoding are outside its scope.
\end{corollary}

\begin{proof}
\emph{Redundancy and deletion correction.}
Use the binary construction of Ye et al.\ with burst length $k$ and
regularity parameter $d\eqdef7$~\cite[Theorem~IV.1]{YeSunYuGeElishco2026}.
We retain the hypotheses $k\geq2$, $n>2k$, and $n\geq9$ of the cited binary
theorem ($d=7$).
In their array convention, append $km-n$ zeros, where
$m\eqdef\lceil n/k\rceil$.  Deleting two bursts from the original word
leaves this suffix intact and uses the same padding at length $n-2k$.
Since $n\leq km<n+k$, the padding changes the redundancy estimate only by
$O_k(1)$.
The array has $k$ rows of length $m$, one for each residue class of
coordinates modulo $k$, with positions $1,\ldots,m$; its first row consists of $m$ original bits.  A row is $d$-regular when each of its factors of length at least
$d\log_2 m$ contains both $00$ and $11$~\cite[Definition~IV.1]{YeSunYuGeElishco2026};
such rows are the input class of the first-row
decoder~\cite[Lemma~IV.2]{YeSunYuGeElishco2026}.
For $3\leq m\leq8$ every row of length $m$ is $d$-regular, since $d\log_2 m>m$
leaves no such factor.
For $m\geq9$ at least half of these rows are regular
by~\cite[Lemma~IV.1]{YeSunYuGeElishco2026}.  Thus at least $2^{n-1}$
original words have a regular first row.

We use the following convention for the row windows in
\cite[Sec.~IV-C]{YeSunYuGeElishco2026}.  Partition $\{1,\ldots,m\}$ into
consecutive blocks of size $P\eqdef\lceil3d\log_2 m\rceil+1$, truncating
the last block at $m$.  Use the unions of two successive blocks as windows,
or the whole row when there is only one block, so that the definition also
covers a row shorter than one block.
In the localized case of Ye et al.~\cite[Lemmas~IV.3 and~IV.5]{YeSunYuGeElishco2026},
their first-row decoder identifies an interval $[c_1,c_2]$ of length at
most $3d\log_2 m$ containing both deletion positions in that row.  The
affected positions in every later row lie in $[c_1-1,c_2]\cap[1,m]$
\cite[Observation~II.2]{YeSunYuGeElishco2026}.  This interval has length
at most $P$, so it lies in some window; fix one, the chosen window.
Index the windows by their first block.  For each window $J$, write the
binary two-deletion syndrome~\cite[Lemma~II.4]{YeSunYuGeElishco2026} of each of
rows $2$ through $k$ on $J$ in an $L$-bit field, where $L$ is the maximum output
width over the window lengths used.  Concatenating the $k-1$ fields gives an
integer $0\leq g_J<2^{(k-1)L}$.  The syndrome tuple contains the sum $h^{(0)}$ of
the $g_J$ over the windows with even index and the corresponding sum $h^{(1)}$
over the windows with odd index; both parity sums are taken modulo
$2^{(k-1)L}$~\cite[Eqs.~(16)--(18)]{YeSunYuGeElishco2026}.  Every window of the
chosen window's parity other than the chosen one is disjoint from it.  The
symbols of such a window are therefore known from the received rows, and
subtracting their checks from the corresponding sum recovers $g_J$ for the
chosen window, as in~\cite[Lemma~IV.7]{YeSunYuGeElishco2026}.  Since every
window has length at most $2P$, $L\leq7\log_2(2P)+O(1)$, so the two parity sums
use at most $14(k-1)\log_2\log_2 n+O_k(1)$ bits, as in the count
of~\cite[Theorem~IV.1]{YeSunYuGeElishco2026}.  In the remaining case, the
two first-row deletions lie in runs with at least one intervening run, and
the separated-runs
reconstruction~\cite[Lemma~IV.6]{YeSunYuGeElishco2026} applies to the padded
word, whose length is less than $n+k$.
Partitioning the words with regular first
row by the resulting syndrome tuples, as in their Theorem~IV.1, gives a
class satisfying \eqref{eq:explicit-two-burst-benchmark}, which is the redundancy
bound of that theorem with $b=k$.

\emph{Transfer to insertion bursts and duplications.}
Apply Lemma~\ref{lem:burst-deletion-to-insertion} to this code.
Every length-$k$ PAL or RC duplication inserts one length-$k$ block.

\emph{Codeword decoding with a supplied syndrome tuple.}
The received length determines the number of errors.  With zero errors,
return the received codeword.  Otherwise enumerate the $O(n)$ length-$k$
factors for one error, or the $O(n^2)$ disjoint pairs for two errors,
and delete the chosen factors.  Lemma~\ref{lem:sequential-burst-containment}
ensures that this search includes the transmitted codeword.
For a fixed syndrome tuple, test membership by the first-row regularity
condition and the checks $\eta,f,h^{(0)},h^{(1)}$ with the row windows
above.  Here $\eta$ is the two-deletion syndrome of the regular first
row~\cite[Lemma~IV.2]{YeSunYuGeElishco2026}, and $f$ is the check on the
single-burst syndromes of the column blocks of its
runs~\cite[Sec.~IV-B]{YeSunYuGeElishco2026}.  Their component computations are polynomial-time
for fixed $k$~\cite[Lemmas~II.2, II.4, II.6, and IV.2]{YeSunYuGeElishco2026}.
Thus every candidate check, and hence the full $O(n^2)$ search, takes
polynomial time.
The insertion-correction property guarantees that the set of distinct
codewords returned on a valid channel output is a singleton, although several
deletion witnesses may yield that same codeword.  For fixed $k$, the
syndrome tuple has $O(\log n+k\log\log n)$ bits, so it is
nonuniform advice of logarithmic length.  The same deletion
search decodes PAL-$k$ and RC-$k$ outputs directly.
\end{proof}

\begin{table}[!t]
\centering
\small
\caption{Neighboring coding regimes used in the comparison.  Redundancy is
measured in $q$-ary symbols unless bits are stated explicitly; rows with
different parameter regimes are not directly ordered by strength.}
\label{tab:related-work}
\begin{tabular}{@{}>{\raggedright\arraybackslash}p{0.20\textwidth}
                    >{\raggedright\arraybackslash}p{0.35\textwidth}
                    >{\raggedright\arraybackslash}p{0.31\textwidth}@{}}
\toprule
Channel and regime & Known result in that regime & Relation to the present work \\
\midrule
Tandem, fixed $k,t$
& Optimal leading redundancy term $t\log_q n$ via tandem-specific
  transforms~\cite{KovacevicTan2018,LenzJungerWachterZeh2018};
  sticky-insertion methods at $k=1$~\cite{MahdavifarVardy2017}
& The endpoint-multiplicity argument gives the same leading term $t\log_q n$ in the converse for PAL and RC. \\
Unbounded errors: PAL and RC, even $q$; RC with fixed-point-free complement
& Exact coding capacities; both capacities are zero for $k\geq2$
  ~\cite[Corollaries~13 and~16]{YohananovSchwartz2025}
& The present problem instead fixes a finite error budget $t$. \\
RC, one odd-length error; even $q$, fixed-point-free complement
& A binary
  arbitrary-burst code lifts to an RC code.  Its bit redundancy is
  $\log_2 n+o(1)$ for $k=1$ and at most
  $\log_2 n+(k-1)\log_2\log_2 n+O_k(1)$ for odd $k\geq3$
  ~\cite{BenTolilaSchwartz2022,SchoenyWachterZehGabrysYaakobi2017}
& The present even-$k$ conjugacy is instead a word-space bijection for every
  $q\geq2$ and every involutive complement, including fixed points. \\
RC, fixed $t$, $k=1$, even $q\geq4$, fixed-point-free complement
& Explicit codes with redundancy
  $(2t-1)\log_q n+O(\log_q\log_q n)$, with
  $O(n\,\operatorname{poly}(\log_2 n))$ encoding and decoding
  ~\cite{SunGe2026}
& For $t\geq2$, the present converse has leading coefficient $t$,
  below the coefficient $2t-1$ of these codes.  The case $t=1$ is covered by the
  next row. \\
PAL/RC, one fixed-length error
& Codes for the stronger arbitrary-burst channel give
  $\log_q n+O_{q,k}(1)$ redundancy~\cite{SunLuZhangGe2025}
& The present converse shows optimality up to an additive constant for $t=1$
  (for RC with odd $k$ and a fixed-point-free complement,
  also~\cite[Theorem~18]{BenTolilaSchwartz2022}). \\
Binary PAL/RC, two errors of fixed even length $k\geq2$
& The two-deletion-burst existence bound~\cite[Theorem~III.1]{YeSunYuGeElishco2026}
  gives $4\log_2 n+O_k(1)$ by Lemma~\ref{lem:burst-deletion-to-insertion}.
  The construction of Ye et al.~\cite[Theorem~IV.1]{YeSunYuGeElishco2026}
  gives redundancy at most $5\log_2 n+14k\log_2\log_2 n+O_k(1)$ and
  polynomial-time codeword decoding for a supplied suitable syndrome tuple
  (Corollary~\ref{cor:explicit-two-burst-benchmark}).
& The present PAL/RC-specific bounds, which hold for every $q$, have
  coefficient $4$ for existence,
  with a smaller additive term (Remark~\ref{rem:two-burst-existence}), and
  coefficient $2$ for the converse; neither the existence bound nor the
  benchmark gives an efficient message encoder. \\
\bottomrule
\end{tabular}
\end{table}

The hierarchy below is stated for even $k$, where
Theorem~\ref{thm:even-k-two-sphere} gives the exact sphere used by
Theorem~\ref{thm:greedy-code}, and for binary codes, for which
Corollary~\ref{cor:explicit-two-burst-benchmark} gives the syndrome-parameterized
benchmark.  Its redundancy upper bound has leading coefficient five.
The efficient-code questions below are posed in this binary setting.

For every fixed even $k$, define
\[
  r^{\min}_2(n)\eqdef\min\bigl\{r_2(C): C\subseteq\{0,1\}^n
  \text{ corrects at most two PAL-$k$ duplications}\bigr\}.
\]
For the first row of the binary hierarchy, take $q=t=2$ in
Corollary~\ref{cor:fixed-number-packing}.  The third row is the benchmark
in Corollary~\ref{cor:explicit-two-burst-benchmark}.
\begin{center}
\small
\setlength{\tabcolsep}{5pt}
\begin{tabular}{@{}lll@{}}
\toprule
status & redundancy & coefficient\\
\midrule
universal converse
  & $r_2(C_n)\geq2\log_2 n-(2k+3)-o(1)$ & $2$\\
greedy existence
  & $r_2(C_n)\leq4\log_2 n+O_k(1)$ for some $C_n$ & $4$\\
syndrome-parameterized burst benchmark
  & $r_2(C_n)\leq5\log_2 n+14k\log_2 \log_2 n+O_k(1)$ & $5$\\
\bottomrule
\end{tabular}
\end{center}
The first two rows of the hierarchy imply
\[
  2\leq\liminf_{n\to\infty}\frac{r^{\min}_2(n)}{\log_2 n}
  \leq\limsup_{n\to\infty}\frac{r^{\min}_2(n)}{\log_2 n}\leq4.
\]
These comparisons lead to an existence question and a construction question.
\begin{openproblem}[Existential coefficient]
\label{prob:binary-two-error-existence}
For fixed even $k\geq2$, determine the asymptotic behavior of
$r^{\min}_2(n)/\log_2 n$.  In particular, does
$r^{\min}_2(n)=2\log_2 n+o(\log n)$, matching the converse?
\end{openproblem}

\begin{samepage}
\begin{openproblem}[Explicit efficient construction]
\label{prob:binary-two-error-code}
The existence bound above, with leading coefficient $4$, is nonconstructive.  Construct,
for a fixed even $k\geq2$ and uniformly in $n$, an explicit family of binary
codes correcting at most two unrestricted sequential palindromic duplications
of length $k$.  The construction should select all defining parameters in
polynomial time without nonuniform advice, provide polynomial-time message
encoding and codeword decoding, and have redundancy whose leading $\log_2 n$
coefficient is below $5$.  Ideally, attain
$2\log_2 n+O_k(\log \log n)$ or $2\log_2 n+O_k(1)$ redundancy; the
construction would then transfer through $\Amap_n$ to
reverse-complement duplication.
\end{openproblem}
\end{samepage}

\begin{remark}[Possible duplication-specific route]
\label{rem:duplication-specific-route}
One possible route toward Open Problem~\ref{prob:binary-two-error-code} is to
exploit the palindromic-duplication sphere structure directly in a code,
rather than correct every pair of arbitrary sequential
length-$k$ insertion bursts.  For $k=2$, the edge-tree normal form suggests
labeling classes of histories with the same source and descendant, each grouping
one or more structural shapes, rather than raw operation
histories, or imposing a genuinely joint constraint on the two interleaved
rows formed by the even- and odd-indexed coordinates.
\end{remark}

For general even $k$, grouping the coordinates by their residues modulo
$k$ gives $k$ columns.  Appendix~\ref{app:cyclic-factorization} records the
exact insertion in each column and the cyclic perfect matchings that describe
the source--insertion coupling, providing additional structure for
interleaved constructions.

Appendix~\ref{app:construction-obstructions} records precise limitations of
several natural separated-row and raw-history approaches.

\section{Conclusion}
\label{sec:conclusion}

We established a common finite-error converse for palindromic and
reverse-complement duplication.  The endpoint-multiplicity argument gives
finite-length bounds on code size and asymptotics for fixed and sublinear
numbers of errors at every duplication length.  These bounds hold in the full sequential channel, in which a
later source block may reuse symbols inserted earlier, but their proof needs to
count only the endpoint-multiplicity family, built from original source occurrences.  For multiple
PAL duplications, this finite-length upper bound advances the question on code
size posed by Lenz, Wachter-Zeh, and Yaakobi; for
fixed $t$ it gives $r_q(C)\geq t\log_q n-O_{q,k,t}(1)$.  At every even
duplication length, the alternating-complement map further gives a bijective
conjugacy of the full sequential channels and the corresponding
transfers of encoders and of codeword and message decoders.

For every even $k$, we also determined the maximum size of an exact-two descendant sphere.
Together with the inverse-parent bound and greedy coloring, this maximum
gives PAL and RC codes with redundancy at most $4\log_q n+O_{q,k}(1)$.  Our
deletion-to-insertion argument transfers the two-deletion-burst existence
bound of Ye et al.\ to insertions, with the same leading order and a larger
additive constant.
For two errors, the converse gives $2\log_q n-O_{q,k}(1)$, a factor of two
below these existence bounds.  For binary codes, the construction for two
arbitrary insertion bursts gives a redundancy upper bound with leading
coefficient five and polynomial-time codeword decoding for a supplied
suitable syndrome tuple.  For $t\geq3$ errors, we prove the converse and, at even $k$, the
conjugacy, but give no construction or existence bound.  For the binary two-error problem at fixed even $k$, Open
Problems~\ref{prob:binary-two-error-existence} and
\ref{prob:binary-two-error-code} ask for the asymptotic behavior of the
minimum redundancy relative to $\log_2 n$ and for a uniform explicit family with
leading coefficient below five and with polynomial-time parameter selection,
message encoding, and codeword decoding.
Remark~\ref{rem:duplication-specific-route} records one possible route toward
the latter target: exploit the palindromic-duplication sphere structure instead of correcting
the strictly stronger arbitrary-burst channel.

\appendix

\section{Cyclic factorization of the even-length interleaving action}
\label{app:cyclic-factorization}

The conjugacy of Theorem~\ref{thm:conjugacy} identifies the PAL and RC
channels at the word level.  The exact law below describes, more finely, how
one duplication couples the $k$ interleaved residue classes.  The packing and sphere results of the paper do not
depend on this refinement.

Fix an even duplication length $k=2m$.  For a word $z$ and
$0\leq r<k$, write
\[
  \operatorname{col}_{k,r}(z)
  \eqdef (z_{r+jk})_{j\geq0,\ r+jk<|z|}
\]
for its residue-$r$ column.  Row indices inside every column start at zero.
We write $\langle d\rangle_k\in\{0,\ldots,k-1\}$ for the least
nonnegative residue of $d$ modulo $k$.
For a word $z=z_0\cdots z_{\ell-1}$, a word $a$, and
$0\leq h<\ell$, let
\[
  \InsAfter_h(z;a)
  \eqdef z_0\cdots z_h\,a\,z_{h+1}\cdots z_{\ell-1};
\]
thus $\InsAfter_h$ inserts $a$ immediately after coordinate $h$.

For integers $b$ and $r$, set
\[
  \pi_b(r)\eqdef\langle 2b-1-r\rangle_k.
\]

\begin{lemma}[Exact even-length interleaving law]
\label{lem:even-k-row-update}
Let $x\in\Sigma^n$, let $0\leq i\leq n-k$, and write
\[
  i=ak+b,
  \qquad 0\leq b<k.
\]
Set $y\eqdef T^{\pal}_{i,k}(x)$ and, for $0\leq r<k$, define
\begin{align*}
  h_b(r)&\eqdef a+\mathbf 1\{r<b\},\\
  s_r&\eqdef i+\langle r-b\rangle_k
      =kh_b(r)+r,\\
  t_r&\eqdef i+k+\langle r-b\rangle_k
      =k\bigl(h_b(r)+1\bigr)+r.
\end{align*}
Then $s_r$ is the unique source-block coordinate in residue class $r$,
$t_r$ is the unique inserted-block coordinate in residue class $r$, and
\begin{equation}
  y_{t_r}=x_{s_{\pi_b(r)}}.
  \label{eq:even-k-source-target-coordinate}
\end{equation}
Equivalently, if $X_r\eqdef\operatorname{col}_{k,r}(x)$ and
$Y_r\eqdef\operatorname{col}_{k,r}(y)$, and $X_r[h]$ denotes the symbol at row
$h$ of column $X_r$, then
\begin{equation}
  Y_r=
  \InsAfter_{h_b(r)}
  \left(
    X_r;
    X_{\pi_b(r)}[h_b(\pi_b(r))]
  \right).
  \label{eq:even-k-column-update}
\end{equation}
\end{lemma}

\begin{proof}
Put $d_r\eqdef\langle r-b\rangle_k$.  The source block occupies coordinates
$i,i+1,\ldots,i+k-1$, so its unique coordinate congruent to $r$ modulo
$k$ is $i+d_r=s_r$.  If $r\geq b$, this coordinate is $ak+r$; if
$r<b$, it is $(a+1)k+r$.  This proves
$s_r=kh_b(r)+r$.  The coordinate of the inserted block with the same
residue is $i+k+d_r=t_r=s_r+k$.

The symbol placed at $t_r$ comes from source coordinate
$i+k-1-d_r$.  Since $0\leq d_r<k$,
\[
  \bigl\langle\pi_b(r)-b\bigr\rangle_k
  =\langle b-1-r\rangle_k
  =k-1-d_r.
\]
Consequently $i+k-1-d_r=s_{\pi_b(r)}$, proving
\eqref{eq:even-k-source-target-coordinate}.  Finally, insertion of a block
of length $k$ preserves the residue class of every original suffix
coordinate and shifts it down by one row in its column.  Thus column $r$
receives one symbol immediately after row $h_b(r)$, with the value
shown in \eqref{eq:even-k-column-update}.
\end{proof}

\begin{remark}[RC column update]
The maps $\pi_b$ and $h_b$ are the same for an RC duplication.  Namely,
if $y^{\rc}\eqdef T^{\rc}_{i,k}(x)$ and
$Y_r^{\rc}\eqdef\operatorname{col}_{k,r}(y^{\rc})$, then
\[
  Y_r^{\rc}=
  \InsAfter_{h_b(r)}
  \left(
    X_r;
    \iota\!\left(X_{\pi_b(r)}[h_b(\pi_b(r))]\right)
  \right).
\]
Thus both channels have the same source--insertion column
coupling; only the inserted symbol value is complemented
in the RC channel.
\end{remark}

Let
\[
  V_0\eqdef\{0,2,\ldots,k-2\},
  \qquad
  V_1\eqdef\{1,3,\ldots,k-1\}.
\]

\begin{proposition}[Cyclic $1$-factorization of the column coupling]
\label{prop:even-k-one-factorization}
For every integer $b$, the map $\pi_b:\Z_k\to\Z_k$ is a
fixed-point-free involution that exchanges $V_0$ and $V_1$.  Moreover,
$\pi_{b+m}=\pi_b$.  Hence, for $c\in\Z_m$, the set
\[
  F_c\eqdef
  \bigl\{\{e,\pi_b(e)\}:e\in V_0\bigr\},
  \qquad b\equiv c\pmod m,
\]
is well defined, and the $m$ matchings $(F_c)_{c\in\Z_m}$ form a
$1$-factorization of $K_{m,m}$ with bipartition $(V_0,V_1)$.
\end{proposition}

\begin{proof}
For every $r\in\Z_k$,
\[
  \pi_b(\pi_b(r))
  \equiv 2b-1-(2b-1-r)
  \equiv r\pmod k,
\]
so $\pi_b$ is an involution.  A fixed point would satisfy
$2r\equiv2b-1\pmod{2m}$, which is impossible.  Also
$\pi_b(r)\equiv1-r\pmod2$, so the involution exchanges even and odd
residues and induces a perfect matching.

The identity
$\pi_{b+m}(r)\equiv\pi_b(r)\pmod{2m}$ shows that $F_c$ is well defined.
Fix $e\in V_0$ and $o\in V_1$.  The edge $\{e,o\}$ lies in $F_c$
precisely when
\[
  2b\equiv e+o+1\pmod{2m}.
\]
The right-hand side is even, and division by two gives the unique solution
\[
  b\equiv\frac{e+o+1}{2}\pmod m.
\]
Thus every edge of $K_{m,m}$ belongs to exactly one $F_c$.
\end{proof}

For $k=2m$, the possible source--insertion column couplings are therefore the
factors in the standard cyclic $1$-factorization of $K_{m,m}$.  The residue
$b\bmod m$ is the matching phase, while the profile $(h_b(r))_{r\in\Z_k}$ is
the row cut and distinguishes the two lifts $b$ and $b+m$ modulo $k$.  These
are two pieces of information that a future dedicated decoder may exploit.

After writing $e=2u$ and $o=2v+1$, the edge condition becomes
$u+v\equiv b-1\pmod m$.  Thus the factorization is the classical cyclic
one, while Proposition~\ref{prop:even-k-one-factorization} identifies its
factors with the source--insertion couplings of the channel.

\begin{corollary}[Identification of the matching phase]
\label{cor:even-k-phase-identification}
Suppose one source--insertion column pair is identified: a symbol inserted
in column $r$ was copied from column $s$.  Then
\[
  b\equiv\frac{r+s+1}{2}\pmod m.
\]
Consequently, only the two lifts $b$ and $b+m$ modulo $k$ remain.
\end{corollary}

\begin{proof}
Lemma~\ref{lem:even-k-row-update} gives $s=\pi_b(r)$, so
$2b\equiv r+s+1\pmod{2m}$.  Since $r$ and $s$ have opposite parity,
division by two determines $b$ uniquely modulo $m$.
\end{proof}

Recovering the complete duplication start therefore also requires the row
location and the choice between the lifts $b$ and $b+m$.

\section{Obstructions for separated-row and raw-history constructions}
\label{app:construction-obstructions}

Two natural first approaches to the open construction problem in
Section~\ref{sec:target} are to separate a PAL2 word into its even and odd rows and
to assign additive weights to individual error histories.  This appendix
proves obstructions for both approaches.  The row analysis first gives pointwise list bounds and then
exhibits a clique of linear size whose rows have bounded run and
alternating-substring lengths; the final result
shows that a raw $B_2$ sum on history labels cannot equal a difference of word syndromes.
These are architecture-specific obstructions, not redundancy lower bounds
for unrestricted PAL2 or RC2 codes.  In particular, they leave open invariants defined on classes of histories
with the same source and descendant, and joint constraints on the two rows.

\subsection{Row dynamics and trace-conditioned reconstruction}

We first determine which information about the even-row history suffices to
reconstruct the odd row.  For a binary PAL2 source
$u=u_0\cdots u_{n-1}\in\{0,1\}^n$, define its even and odd rows by
\[
  E=E(u)\eqdef (u_{2j})_{0\leq j<\lceil n/2\rceil},
  \qquad
  O=O(u)\eqdef (u_{2j+1})_{0\leq j<\lfloor n/2\rfloor}.
\]
In the notation of Appendix~\ref{app:cyclic-factorization}, $E$ and $O$ are
the residue-class columns $\operatorname{col}_{2,0}(u)$ and
$\operatorname{col}_{2,1}(u)$, and $\InsAfter_h(z;a)$ inserts $a$ immediately
after coordinate $h$ of $z$.  The next lemma is the case $k=2$ of
Lemma~\ref{lem:even-k-row-update}.

\begin{lemma}[Exact row update]
\label{lem:row-update}
Let $0\leq i\leq n-2$, write $i=2j+\delta$ with
$\delta\in\{0,1\}$, and let $u'\eqdef T^{\pal}_{i,2}(u)$.  If $E',O'$ are the even
and odd rows of $u'$, then
\[
  E'=\InsAfter_{j+\delta}(E;O_j),
  \qquad
  O'=\InsAfter_j(O;E_{j+\delta}).
\]
The formulas are applied to the updated rows after every subsequent error.
\end{lemma}

\begin{proof}
Read the even and odd coordinates of the local replacement
$ab\mapsto abba$.  The two choices of $\delta$ determine which of $a,b$ lies
in each row and give the displayed insertion coordinates.
\end{proof}

\begin{proposition}[The original even row is insufficient]
\label{prop:even-insufficient}
There is no decoder that, for every one-error PAL2 instance, recovers the
source from the pair consisting of its original even row and the full
received word.
\end{proposition}

\begin{proof}
The two direct calculations
\[
  001000\xrightarrow{\,i=1\,}00110000,
  \qquad
  001100\xrightarrow{\,i=4\,}00110000
\]
have the same original even row $E=010$, while their odd rows are $000$ and
$010$, respectively.  Thus the same side information and received word have
two distinct sources.  By conjugacy, the same obstruction holds for RC2,
since $\Amap_n$ fixes the even coordinates and complements the odd ones.
\end{proof}

This ambiguity concerns the stated side information: a joint syndrome of
$(E,O)$ may separate the two sources.

\begin{lemma}[Conditional trace uniqueness]
\label{lem:trace-unique}
Fix any number $t$ of PAL2 errors.  Suppose the complete ordered even-row
trace is known: every intermediate even row and every insertion coordinate
is specified.  Together with the final odd row, these data admit at most one
compatible initial odd row.
\end{lemma}

\begin{proof}
Reverse one step.  Suppose the new even-row symbol was inserted at coordinate
$r\geq1$.  Let $a$ be the symbol at coordinate $r-1$ of the even row before
that insertion.  The symbol inserted into the odd row in the same step is $a$.  Depending
on whether the forward start had $\delta=0$ or $\delta=1$, the reverse step deletes
$a$ at odd-row coordinate $r$ or $r-1$ (the latter only when legal).  If only
one coordinate contains $a$, the deletion is forced.  If both do, the two
adjacent symbols are equal, and deleting either produces the same predecessor
word.  Boundary and inserted-value checks may reject a candidate but cannot
create another.  Induction over the reversed trace proves uniqueness.
\end{proof}

\subsection{Pointwise list bounds and linear fixed-row ambiguity}

A known even row can leave several sources compatible with one received
word.  A run-length bound controls this list, although linear ambiguity can
persist.  Suppose $|E|=m$ and every run of $E$ has length at most $L$.  After two
PAL2 errors, the received even row has length $m+2$.

\begin{lemma}[Deletion multiplicity]
\label{lem:deletion-multiplicity}
If deleting a coordinate of a word $w$ produces a fixed word $E$, then all
coordinates whose deletion produces $E$ form one contiguous run of equal
symbols in $w$.  If every run of $E$ has length at most $L$, their number is at most $L+1$.
\end{lemma}

\begin{proof}
Two different deletions give the same word exactly when all symbols between
the deleted coordinates are equal.  Deleting one symbol from such a run of
length $R$ leaves a run of length at least $R-1$ in $E$, so $R\leq L+1$.
\end{proof}

\begin{theorem}[A two-error list bound]
\label{thm:rll-list-upper}
For a fixed original even row $E$ of length $m$ whose runs have length at most $L$,
and a fixed received word $y$, the number
of compatible initial odd rows under two PAL2 errors is at most
\[
  (m+2)(L+1).
\]
\end{theorem}

\begin{proof}
In reverse, choose one of at most $m+2$ first deletion coordinates in the
received even row.  In the intermediate row,
Lemma~\ref{lem:deletion-multiplicity} leaves at most $L+1$ second deletion
coordinates that return $E$.  This already counts both deletion orders.  For
each ordered even-row trace, Lemma~\ref{lem:trace-unique} permits at most one
initial odd row.
\end{proof}

The dependence on $m$ cannot be removed, even when every run of $E$
has length one.

For equal-length words $a$ and $b$, write
\[
  \weave(a,b)\eqdef a_0b_0a_1b_1\cdots
\]
for their coordinate-wise interleaving.

\begin{theorem}[Linear ambiguity at maximum run length one]
\label{thm:rll-obstruction}
For infinitely many source lengths $n$, there are a fixed even row $E$ with
maximum run length one, a fixed received word $y$, and at least
$(n+2)/4$ distinct initial odd rows whose woven sources all reach $y$ after
two PAL2 errors.
\end{theorem}

\begin{proof}
For $r\geq1$, define
\[
  E_r\eqdef(01)^{2r-1}0,
  \qquad \widehat E_r\eqdef(01)^{2r}0,
  \qquad Q_r\eqdef(0100)^r0.
\]
Put
\[
  y_r\eqdef\weave(\widehat E_r,Q_r)=(00110010)^r00.
\]
For $0\leq s<r$, let $Q^A_{r,s}$ be obtained from $Q_r$ by deleting the
block $01$ at coordinates $4s,4s+1$, and let $Q^B_{r,s}$ be obtained by
deleting the block $00$ at coordinates $4s+2,4s+3$.  Define
\[
  u^A_{r,s}\eqdef\weave(E_r,Q^A_{r,s}),
  \qquad
  u^B_{r,s}\eqdef\weave(E_r,Q^B_{r,s}).
\]
Applying Lemma~\ref{lem:row-update} twice gives
\begin{align*}
  u^A_{r,s}&\xrightarrow{\,8s+1\,}\cdot
             \xrightarrow{\,8s+3\,}y_r,\\
  u^B_{r,s}&\xrightarrow{\,8s+4\,}\cdot
             \xrightarrow{\,8s+3\,}y_r.
\end{align*}
Indeed, in each case the two row insertions restore $\widehat E_r$ and
$Q_r$.  The $A$-family deletes one of the $r$ ones of $Q_r$; the $B$-family
retains all of those ones and shortens a different zero gap.  Hence all $2r$
sources are distinct.  Their common length is $n=8r-2$, while $E_r$ has
maximum run length one.  Therefore $2r=(n+2)/4$ distinct odd rows share the
same $E_r$ and $y_r$.
\end{proof}

\begin{corollary}[Conditional cost in a separated-row architecture]
\label{cor:rll-architecture}
Consider a separated-row construction whose constraints other than a side
syndrome on the odd row of the source admit all $2r$ sources in the family of
Theorem~\ref{thm:rll-obstruction}.  If every value class of that side
syndrome must correct two PAL2 errors within this otherwise admissible
fixed-$E_r$ family, then the syndrome needs at least $2r=(n+2)/4$ values,
or $\log_2 n-O(1)$ bits.  In particular, when $L=1$, no side syndrome with
$\operatorname{poly}(L)=O(1)$ values meets this requirement.
\end{corollary}

\begin{proof}
All $2r$ otherwise admissible sources have the same exact-two descendant
$y_r$.  No two can therefore belong to the same error-correcting syndrome
class.
\end{proof}

\subsection{Regularity localization}

Weak regularity can localize the candidate traces for one received word.
We first isolate the regularity-free alignment result used in the argument,
then derive the localization and trace-list bounds.  The clique construction
in the following subsection shows why these pointwise bounds do not by
themselves yield a small universal side syndrome.

For an integer $\ell\geq1$, a binary word is called
\emph{weakly $\ell$-regular} if
\begin{enumerate}
  \item every constant run has length at most $\ell$; and
  \item every alternating substring, of the form $0101\cdots$ or
  $1010\cdots$, has length at most $\ell$.
\end{enumerate}
The $d$-regularity condition of the deletion-code literature
implies weak $\lfloor d\log_2 n\rfloor$-regularity of a word of length $n$~\cite[Remark~1 of the preprint version]{SongCai2023}, but not conversely.

For a word $x$ and a word $z$ shorter by two symbols, let
$\mathcal P(x,z)$ be the set of unordered pairs $\{p,q\}$, where
$0\leq p<q<|x|$, such that deleting coordinates $p$ and $q$ from $x$
produces $z$.

We use the following regularity-free result of Song and
Cai~\cite[Lemma~11 of the preprint version]{SongCai2023}.

\begin{lemma}[Two-deletion alignment]
\label{lem:two-deletion-alignment}
Let $x$ be a binary word and $z$ a word shorter by two symbols.  Let
$\{p,q\},\{p',q'\}\in\mathcal P(x,z)$, with $p<q$, $p'<q'$, and, after
interchanging the two pairs if necessary, $p\leq p'$.  Then one of the
following holds:
\begin{enumerate}
  \item $p,p'$ lie in one run of $x$, and $q,q'$ lie in one run of $x$;
  or
  \item there is an alternating substring $x_a\cdots x_b$ of length at
  least three such that $p$ and $a$ lie in one run, $q=a+1$,
  $p'=b-1$, and $q'$ and $b$ lie in one run.
\end{enumerate}
\end{lemma}

Now let $E_y$ be obtained from $E$ by two insertions.  Equivalently,
$\mathcal P(E_y,E)$ is nonempty.

\begin{lemma}[Regularity after two insertions]
\label{lem:regularity-transfer}
If $E$ is weakly $\ell$-regular, then every constant run of $E_y$ has length
at most $\ell+2$, and every alternating substring of $E_y$ has length at
most $3\ell+2$.
\end{lemma}

\begin{proof}
Delete from a constant run of $E_y$ all inserted symbols that it contains.
If an original symbol remains, the surviving symbols form a contiguous
constant substring of $E$; otherwise the run has length at most two.  In
either case, the run has length at most $\ell+2$.

Likewise, deleting the inserted symbols from an alternating substring of
$E_y$ partitions the surviving original symbols into at most three
contiguous alternating substrings of $E$.  Each has length at most $\ell$, and
the two deleted symbols contribute at most two more positions.  The alternating
substring of $E_y$ therefore has length at most $3\ell+2$.
\end{proof}

\begin{corollary}[Localization under weak regularity]
\label{cor:regularity-localization}
Suppose that $m\geq1$, $E\in\{0,1\}^m$ is weakly $\ell$-regular, and
$E_y\in\{0,1\}^{m+2}$ is obtained from $E$ by two insertions.  The pair
$(E,E_y)$ determines one of the following alternatives:
\begin{enumerate}
  \item two distinct runs $\mathcal R_1,\mathcal R_2$ of $E_y$ such that
  every pair of deletions taking $E_y$ to $E$ removes one symbol from each
  run; or
  \item an interval $J\subseteq\{0,\ldots,m+1\}$ containing both deletion
  coordinates of every pair taking $E_y$ to $E$, with
  \[
    |J|\leq 5\ell+6.
  \]
\end{enumerate}
\end{corollary}

\begin{proof}
Lemma~\ref{lem:regularity-transfer} shows that every run of $E_y$ has length
at most $\ell+2$ and every alternating substring has length at most
$3\ell+2$.

Choose the lexicographically first pair $\{p,q\}\in\mathcal P(E_y,E)$;
this choice is determined by $(E,E_y)$.  If $p$ and $q$ lie in one run
$\mathcal R$, the second alternative of
Lemma~\ref{lem:two-deletion-alignment} is impossible, because it places the
two coordinates of each deletion pair in different runs.  The first
alternative therefore confines every pair in $\mathcal P(E_y,E)$ to
$\mathcal R$.  Taking $J\eqdef\mathcal R$ gives the interval alternative.

Suppose next that $p$ and $q$ lie in distinct runs
$\mathcal R_1,\mathcal R_2$.  In the second alignment alternative, the two
coordinates of either deletion pair lie in adjacent runs, at least one of
which is a singleton.  Indeed, since the alternating substring has length at
least three, $q=a+1$ is flanked within it by equal symbols; similarly,
$p'=b-1$ is flanked within it by equal symbols.  Hence, if
$\mathcal R_1,\mathcal R_2$ are nonadjacent, or if they are adjacent and both
have length at least two, only the first alignment alternative can occur.
Every compatible pair then removes one symbol from each of these two runs,
which gives the first alternative of the corollary.

It remains that $\mathcal R_1,\mathcal R_2$ are adjacent and at least one
is a singleton.  Extend across their boundary to the maximal alternating
substring and adjoin the full runs containing its two endpoints.  The
resulting interval $J$ is determined by the chosen pair and hence by
$(E,E_y)$.  Under the first alignment alternative, a compatible pair lies
in $\mathcal R_1\cup\mathcal R_2\subseteq J$.  Under the second, the
lexicographic choice gives $p\leq p'$, so $q=a+1$ and $a=q-1$ is the last
coordinate of $\mathcal R_1$.  Thus $[a,b]$ crosses the chosen boundary and
lies in its maximal alternating substring.  The coordinate $p'=b-1$ lies
inside this substring, while $q'$ lies in the run containing $b$, which is
either a singleton inside the substring or one of its endpoint runs.
Hence every compatible coordinate lies in $J$.  This interval
consists of one alternating substring together with at most two boundary
runs, so
\[
  |J|\leq (\ell+2)+(3\ell+2)+(\ell+2)=5\ell+6.
\]
This proves the second alternative.
\end{proof}

Song
and Cai~\cite[Lemma~9]{SongCai2023}, restated by Ye
et al.~\cite[Lemma~IV.3]{YeSunYuGeElishco2026}, combine the corresponding
localization argument with $d$-regularity of the pre-deletion word.  The
corollary above instead applies the alignment lemma to the pre-deletion word
$E_y$ and obtains its required run and alternating-substring bounds from
weak $\ell$-regularity of the shorter word $E$ via
Lemma~\ref{lem:regularity-transfer}.

\begin{corollary}[Trace and pointwise source-list bounds]
\label{cor:regularity-trace-list}
Under the hypotheses of Corollary~\ref{cor:regularity-localization}, the number of ordered two-insertion traces mapping $E$ to $E_y$ is at most
\[
  (5\ell+6)(5\ell+5).
\]
For fixed $E,E_y$, and final odd row $O_y$, the number of compatible initial
odd rows is bounded by the same quantity.
\end{corollary}

\begin{proof}
Each pair in $\mathcal P(E_y,E)$ has at most two reverse deletion orders,
and hence gives at most two ordered insertion traces.  In the two-run
alternative there are therefore at most
$2|\mathcal R_1||\mathcal R_2|\leq2(\ell+2)^2$ ordered traces.  In the
interval alternative there are at most
$2\binom{|J|}{2}=|J|(|J|-1)$.  The displayed bound dominates both
quantities.

Each ordered even-row trace specifies the intermediate even rows and
insertion coordinates.  For the fixed final odd row $O_y$,
Lemma~\ref{lem:trace-unique} therefore permits at most one initial odd row
per trace.
\end{proof}

For $\ell=O(\log n)$, the corollary gives an $O(\ell^2)$ list for each
received word, but it does not bound the number of syndrome values needed
over the union of those lists.

\subsection{A regular fixed-row clique}

A clique with a common even row forces distinct side-syndrome values across
its vertices, even when different pairs collide at different received words.
Recall that $T^{\pal}_{i,2}$ replaces $ab$ locally by $abba$, and that
$\weave(A,B)$ denotes the coordinate-wise interleaving of two equal-length
rows $A$ and $B$.

\begin{theorem}[A fixed-even-row regular clique]
\label{thm:regular-fixed-row-clique}
For every $r\geq1$, there are $r+1$ binary sources of full length $n=24r$
that share one even row and form a clique in the confusability graph for
exactly two PAL2 duplications.  The common even row and every odd row in the
family have constant-run and alternating-substring lengths bounded by an
absolute constant.
\end{theorem}

\begin{proof}
\emph{Regular source family.}
Let the row length be $m\eqdef12r$ and define
\[
  E\eqdef(0011)^{3r},
  \qquad
  O_*\eqdef(111000)^{2r}.
\]
For $0\leq s<r$, put $j_s\eqdef12s+4$, and obtain $O_s$ from $O_*$ by changing
the symbol at coordinate $j_s$ from zero to one.  Define
\[
  U_*\eqdef\weave(E,O_*),
  \qquad
  U_s\eqdef\weave(E,O_s).
\]
The maximum constant-run and alternating-substring lengths are both two in
$E$; they are three and two, respectively, in $O_*$; and they are at most
three and five, respectively, in every $O_s$.

\emph{Local collision identity.}
The construction is based on the local identity
\begin{equation}
  T^{\pal}_{0,2}(100000)=10010000=T^{\pal}_{4,2}(100100).
  \label{eq:regular-gadget}
\end{equation}
Around coordinate $j_s$, the corresponding six-symbol woven factor is
$100000$ in $U_*$ and $100100$ in $U_s$.  Its global starts in
\eqref{eq:regular-gadget} are
\[
  a_s\eqdef2j_s-2=24s+6,
  \qquad
  b_s\eqdef a_s+4=24s+10.
\]
These six-symbol gadgets are pairwise disjoint and are disjoint from the length-two factor at start
zero.

\emph{Pairwise confusability.}
For the pair $U_*,U_s$, apply the appropriate gadget operation first and
then the common operation at start zero.  The respective ordered histories are
\[
  [a_s,0]
  \qquad\text{and}\qquad
  [b_s,0].
\]
Identity~\eqref{eq:regular-gadget} makes the intermediate words equal at the
only initially different site, and the second operation is common.  Thus the
two sources have a common exact-two descendant.

Now fix $0\leq s<s'<r$.  In $U_s$, first expand the right-hand gadget at
site $s'$ and then the left-hand gadget at site $s$; in $U_{s'}$, do the
corresponding two operations with the opposite local variants.  The histories are
\[
  [a_{s'},b_s]
  \qquad\text{and}\qquad
  [b_{s'},a_s],
\]
respectively.  Operating on the right-hand site first keeps the left-hand
start unchanged.  In both sources, each of the two sites becomes
$10010000$, so the final descendants coincide.  Every pair among
$\{U_*,U_0,\ldots,U_{r-1}\}$ is therefore confusable after exactly two PAL2
duplications.  The family has size $r+1=n/24+1$ and satisfies the asserted
constant regularity bounds.
\end{proof}

\begin{corollary}[Cost of a universal side syndrome in a fixed row]
\label{cor:regular-side-syndrome}
Fix an architecture that, for every binary even row $E$ of length $m$,
assigns every odd row $O\in\{0,1\}^m$ a side-syndrome value
$h_E(O)$ in an alphabet $\mathcal S_E$, with the requirement that each
value class in this fixed-$E$ fiber correct exactly two PAL2
duplications.  For infinitely many full lengths $n$, some fixed row $E$
requires
\[
  |\mathcal S_E|\geq\frac{n}{24}+1
\]
syndrome values.  Consequently, such a universal side syndrome requires
$\log_2 n-O(1)$ bits in the worst case, even when both rows obey constant
run-length and alternating-substring bounds.
\end{corollary}

\begin{proof}
Apply the architecture to the fixed even row in
Theorem~\ref{thm:regular-fixed-row-clique}.  The $n/24+1$ odd rows form a
clique, so no two may receive the same syndrome value.  Taking binary
logarithms gives the stated bit requirement.
\end{proof}

\begin{remark}[Scope of the separated-row obstruction]
For binary words of even length,
\[
  \Amap_n\bigl(\weave(E,O)\bigr)=\weave(E,\overline O).
\]
Consequently, Theorem~\ref{thm:regular-fixed-row-clique} and
Corollary~\ref{cor:regular-side-syndrome} transfer verbatim to exact-two
RC2: the even row remains fixed, while complementation of the odd row
preserves its constant-run and alternating-substring bounds.  These results
are not redundancy lower bounds for unrestricted PAL2 or RC2.  They bound
the cost of coloring the full fixed-$E$ fiber by a universal side syndrome
on $O$, with $E$ regarded as given for free; any redundancy needed to
constrain or recover $E$ is outside this cost, and a joint constraint on
$(E,O)$ may discard the local gadgets.
\end{remark}

\subsection{Raw history labels are not endpoint syndromes}

We show that a $B_2$ sum of two raw history labels, defined below, cannot be
represented as a difference of word syndromes of the descendant and the source.  Recall from Section~\ref{sec:edge-tree} the
$4n-5$ formal adjacency labels
\[
  \cU_n=
  \{P_i:0\leq i\leq n-2\}
  \sqcup
  \{L_i,M_i:0\leq i\leq n-2\}
  \sqcup
  \{R_i:0\leq i\leq n-3\}.
\]
As in the proof of Lemma~\ref{lem:normal-form}, a history whose first
operation is at start $i$ and whose second operation is at dynamic start
$h$ has the raw labels $P_i$ and, respectively, $P_h$ if $h\leq i$, $L_i$,
$M_i$ or $R_i$ if $h=i+1,i+2,i+3$, and $P_{h-2}$ if $h\geq i+4$.
Let $\Gamma$ be a common abelian group.  A $B_2$ weighting
$w:\cU_n\to\Gamma$ means that the sums $w(\lambda)+w(\mu)$ are
distinct for distinct unordered pairs $\{\lambda,\mu\}$, with repetition
allowed.  Consider length-specific word syndromes
$\sigma_m:\{0,1\}^m\to\Gamma$ whose endpoint difference is required to equal
the corresponding raw history sum.

\begin{proposition}[Raw shape sums cannot be endpoint differences]
\label{prop:raw-b2}
For $n\geq2$, there are no maps $\sigma_n:\{0,1\}^n\to\Gamma$,
$\sigma_{n+4}:\{0,1\}^{n+4}\to\Gamma$, and a $B_2$ weighting
$w:\cU_n\to\Gamma$ such that every two-step history from $u$ to $y$ with
raw structural labels $\lambda,\mu$ satisfies
\[
  \sigma_{n+4}(y)-\sigma_n(u)=w(\lambda)+w(\mu)
  \qquad\text{in }\Gamma.
\]
\end{proposition}

\begin{proof}
Starting from $u\eqdef0^n$, the first PAL2 operation at start zero produces
$0^{n+2}$.  A second operation at dynamic start zero or one produces the
same endpoint $y=0^{n+4}$:
\[
  0^n\xrightarrow{\,0\,}0^{n+2}\xrightarrow{\,0\,}0^{n+4},
  \qquad
  0^n\xrightarrow{\,0\,}0^{n+2}\xrightarrow{\,1\,}0^{n+4}.
\]
The raw shapes are $\{P_0,P_0\}$ and $\{P_0,L_0\}$, respectively.  Their
$B_2$ sums in $\Gamma$ must differ, whereas both would equal the same endpoint
difference
$\sigma_{n+4}(0^{n+4})-\sigma_n(0^n)$, a contradiction.
\end{proof}

\begin{remark}[Scope of the raw-history obstruction]
Proposition~\ref{prop:raw-b2} concerns weights on raw history labels; it does
not rule out additive labels on classes of histories with the same source and descendant or a
state-dependent invariant that telescopes along the updated word.
\end{remark}

\bibliographystyle{plainnat}
\bibliography{references}

\end{document}